\documentclass[11pt,reqno,nosumlimits]{amsart}
\usepackage{graphicx, enumerate, url}
\usepackage[margin=1in]{geometry}
\usepackage{amssymb,mathtools}
\usepackage{color}

\newtheorem{theorem}{Theorem}[section]

\newtheorem{proposition}[theorem]{Proposition}
\newtheorem{lemma}[theorem]{Lemma}
\newtheorem{corollary}[theorem]{Corollary}

\theoremstyle{definition}
\newtheorem{definition}[theorem]{Definition}

\theoremstyle{remark}

\begin{document}
\title{Nuclear Norm of Higher-Order Tensors}
\author[S.~Friedland]{Shmuel~Friedland}
\address{Department of Mathematics, Statistics and Computer Science,  University of Illinois, Chicago}
\email{friedlan@uic.edu}
\author[L.-H.~Lim]{Lek-Heng~Lim}
\address{Computational and Applied Mathematics Initiative, Department of Statistics,
University of Chicago}
\email{lekheng@galton.uchicago.edu}
\begin{abstract}
We establish several mathematical and computational properties of the nuclear norm for higher-order tensors. We show that like tensor rank, tensor nuclear norm is dependent on the choice of base field --- the value of the nuclear norm of a real $3$-tensor depends on whether we regard it as a real $3$-tensor or a complex $3$-tensor with real entries.  We show that every tensor has  a nuclear norm attaining decomposition and every symmetric tensor has a symmetric nuclear norm attaining decomposition. There is a corresponding notion of nuclear rank that, unlike tensor rank, is upper semicontinuous. We establish an analogue of Banach's theorem for tensor spectral norm and Comon's conjecture for tensor rank --- for a symmetric tensor, its symmetric nuclear norm always equals its nuclear norm. We show that computing tensor nuclear norm is  NP-hard in several sense. Deciding weak membership in the nuclear norm unit ball of $3$-tensors is NP-hard, as is finding an $\varepsilon$-approximation of nuclear norm for $3$-tensors.  In addition, the problem of computing spectral or nuclear norm of a $4$-tensor is NP-hard, even if we restrict the $4$-tensor to be bi-Hermitian, bisymmetric, positive semidefinite, nonnegative valued, or all of the above. We discuss some simple polynomial-time approximation bounds. As an aside, we show that the nuclear $(p,q)$-norm of a matrix is NP-hard in general but can be computed in polynomial-time if $p=1$, $q = 1$, or $p=q=2$, with closed-form expressions for the nuclear $(1,q)$- and $(p,1)$-norms.
\end{abstract}
\maketitle

\section{Introduction}

The nuclear norm of a $2$-tensor (or, in coordinate form, a matrix) has recently found widespread use as a convex surrogate for rank, allowing one to relax various intractable rank minimization problems into tractable convex optimization problems. More generally, for $\mathbb{F} = \mathbb{R}$ or $\mathbb{C}$, the nuclear norm of a $d$-tensor  $A\in \mathbb{F}^{n_1} \otimes \dots \otimes \mathbb{F}^{n_d} = \mathbb{F}^{n_1 \times \dots \times n_d}$ is defined by
\begin{equation}\label{eq:intro}
\|A\|_{*,\mathbb{F}} = \inf\Bigl\{\sum_{i=1}^r \lvert \lambda_i\rvert : A =  \sum_{i=1}^r \lambda_i u_{1,i}\otimes \dots \otimes u_{d,i}, \; \lVert u_{k,i} \rVert = 1, \; r \in \mathbb{N}\Bigr\}
\end{equation}
where $\|\cdot\|$ is the $l^2$-norm and $u_{k,i} \in \mathbb{F}^{n_k}$ for $k=1,\dots,d$, $i=1,\dots,r$. The nuclear norm of a matrix is then the case when $d = 2$ and is equivalent to the usual definition as a sum of singular values, also known as the Schatten $1$-norm \cite{Book}. For higher-order tensors it was defined explicitly in \cite{LC1, LC2} (see also \cite{Derk, Fri83}) although the original idea dates back to Grothendieck \cite{Gro} and Schatten \cite{Scha}.  In Section~\ref{sec:norms} we will discuss the definitions and basic properties of Hilbert--Schmidt, spectral, and nuclear norms for tensors of arbitrary orders over $\mathbb{C}$ and $\mathbb{R}$ as well as their relations with the projective and injective norms in operator theory.

\subsection{Mathematical properties of tensor nuclear norm}

We start by showing in Section~\ref{sec:special} that the expression in \eqref{eq:intro} defines a norm  and that the infimum is always attained, i.e., there is a finite $r$ and a decomposition into a linear combination of $r$ norm-one rank-one terms such that the $l^1$-norm of the $r$ coefficients gives the nuclear norm. We call this a \emph{nuclear decomposition}. Such a decomposition gives a corresponding notion of \textit{nuclear rank} that, unlike the usual tensor rank, is upper semicontinuous and thus avoids the ill-posedness issues in the best rank-$r$ approximation problem for tensor rank \cite{DSL}. As an aside, we show that one cannot get a Schatten $p$-norm for tensors in this manner: If the $l^1$-norm of the coefficients is replaced by an $l^p$-norm for any $p > 1$, the infimum is identically zero. In Section~\ref{sec:decomp}, we give a necessary and sufficient condition for checking whether a given decomposition of a tensor into rank-one terms is a nuclear decomposition of that tensor. We also show that every norm on a real finite-dimensional vector space may be regarded as a nuclear norm in an appropriate sense.

For notational simplicity let $d = 3$ but the following conjecture and results may be stated for any $d \ge 3$. Let $A\in \mathsf{S}^{3}(\mathbb{F}^{n})$ be a symmetric tensor. Comon's conjecture \cite{CGLM} asserts that the rank and symmetric rank of $A$ are always equal, i.e.,
\begin{equation}\label{eq:rk}
\min \Bigl\{ r : A = \sum_{i=1}^{r}\lambda_{i}u_{i}\otimes v_{i}\otimes w_{i}\Bigr\} \overset{?}{=} 
\min \Bigl\{ r : A = \sum_{i=1}^{r}\lambda_{i}v_{i}\otimes v_{i}\otimes v_{i}\Bigr\}.
\end{equation}
Banach's theorem \cite{Ban38,Fri13} on the other hand shows that the analogous statement holds for the spectral norm in place of rank, i.e.,
\[
\sup_{x,y,z\ne0}\frac{\lvert \langle A, x\otimes y\otimes z\rangle \rvert}{\lVert x\rVert\lVert y\rVert\lVert z\rVert} =
\sup_{x\ne0}\frac{\lvert \langle A, x\otimes x \otimes x\rangle\rvert}{\lVert x\rVert^3} .
\]
We prove the analogous statement for nuclear norm (for arbitrary $d$) in Section~\ref{sec:banach}:
\begin{equation}\label{eq:nn}
\inf\Bigl\{  \sum_{i=1}^{r}\lvert\lambda_{i} \rvert :A=\sum_{i=1}^{r}\lambda_{i}u_{i}\otimes v_{i}\otimes w_{i}\Bigr\} 
= \inf\Bigl\{  \sum_{i=1}^{r}\lvert\lambda_{i} \rvert :A=\sum_{i=1}^{r}\lambda_{i}v_{i}\otimes v_{i}\otimes v_{i}\Bigr\},
\end{equation}
where the infimum is taken over all   $r \in \mathbb{N}$ and $\lVert u_i \rVert = \lVert v_i \rVert = \lVert w_i \rVert = 1$, $i=1,\dots, r$.
This may be viewed as a dual version of Banach's theorem or, if we regard tensor nuclear norm as a continuous proxy for tensor rank, then this shows that the continuous analogue of Comon's conjecture is true.  In addition, we show that every symmetric tensor over $\mathbb{F}$ has a \textit{symmetric nuclear decomposition} over $\mathbb{F}$, i.e., a decomposition that attains the right-hand side of \eqref{eq:nn}.

Tensor rank is known to depend on the choice of base field \cite{Bry, DSL}. We show in Section~\ref{sec:base} that the same is true for  nuclear and spectral norms. If we define $B, C \in \mathbb{R}^{2 \times 2 \times 2} \subseteq \mathbb{C}^{2 \times 2 \times 2}$ by
\begin{align*}
B&= \frac{1}{2}(e_1\otimes e_1\otimes e_2
+ e_1\otimes e_2\otimes e_1
+ e_2\otimes e_1\otimes e_1
- e_2\otimes e_2\otimes e_2),\\
C&= \frac{1}{\sqrt{3}}(e_1\otimes e_1\otimes e_2
+ e_1\otimes e_2\otimes e_1
+ e_2\otimes e_1\otimes e_1),
\end{align*}
where $e_1, e_2 \in \mathbb{R}^2$ are the standard basis vectors,  then
\[
\lVert B\rVert_{\sigma,\mathbb{R}}=1/2 < 1/\sqrt{2} = \lVert B\rVert_{\sigma,\mathbb{C}},\qquad \lVert C\rVert_{*,\mathbb{C}}=3/2  < \sqrt{3} = \lVert C\rVert_{*,\mathbb{R}}.
\]
We give explicit nuclear decompositions and symmetric nuclear decompositions of $B$ and $C$ over $\mathbb{R}$ and $\mathbb{C}$.

As our title indicates, most of this article is about nuclear norms of $d$-tensors where $d \ge 3$. Section~\ref{sec:matrix} is an exception in that it is about the nuclear $(p,q)$-norm for matrices,
\[
\|A\|_{*,p,q} = \inf\Bigl\{\sum_{i=1}^r \lvert \lambda_i\rvert : A =  \sum_{i=1}^r \lambda_i u_i v_i^\mathsf{T}, \; \lVert u_i \rVert_p =  \lVert u_i \rVert_q = 1, \; r \in \mathbb{N}\Bigr\}.
\]
We discuss its computational complexity --- polynomial-time if $p=1$ or $q = 1$ or $p=q=2$, but NP-hard otherwise --- and show that the nuclear $(1,q)$- and $(p,1)$-norms have nice closed-form expressions.

\subsection{Computational properties of tensor nuclear norm}

More generally, we may also define the nuclear $p$-norm of a $d$-tensor  $A\in \mathbb{F}^{n_1 \times \dots \times n_d}$ by
\[
\|A\|_{*,p} = \inf\Bigl\{\sum_{i=1}^r \lvert \lambda_i\rvert : A =  \sum_{i=1}^r \lambda_i u_{1,i}\otimes \dots \otimes u_{d,i}, \; \lVert u_{k,i} \rVert_p = 1, \; r \in \mathbb{N}\Bigr\}
\]
where $\|\cdot\|_p$ is the $l^p$-norm and $u_{k,i} \in \mathbb{F}^{n_k}$ for $k=1,\dots,d$, $i=1,\dots,r$.  When $p = 2$, the nuclear $2$-norm is just the nuclear norm in \eqref{eq:intro}.

For the special case $d = p =2$, the matrix nuclear norm is polynomial-time computable to arbitrary accuracy, as we had pointed out above. Obviously, the computational tractability of the matrix nuclear norm is critical to its recent widespread use. In Sections~\ref{sec:matrix} and \ref{sec:np}, we discuss the computational complexity of the nuclear norm in cases when $p \ne 2$ and $d \ne 2$. We will show that the following norms are all NP-hard to compute:
\begin{enumerate}[\upshape (i)]
\item\label{p2} nuclear $p$-norm of $2$-tensors if $p \ne 1, 2, \infty$,
\item\label{2dR} nuclear $2$-norm of $d$-tensors over $\mathbb{R}$ for all $d \ge 3$,
\item\label{2dC} nuclear $2$-norm of $d$-tensors over $\mathbb{C}$ for all $d \ge 4$.
\end{enumerate}
We rely on our earlier work \cite{FL1} for \eqref{p2} and \eqref{2dR}: The NP-hardness of the nuclear $p$-norm of $2$-tensors follows from that of the operator $p$-norm for $p \ne 1, 2,\infty$ \cite{HO10}; the NP-hardness of the nuclear norm of real $3$-tensors follows from that of the spectral norm of real $3$-tensors \cite{HL13}.

For \eqref{2dC}, we establish a stronger result --- we show that even if we require our $4$-tensor to be bi-Hermitian, bisymmetric, positive semidefinite, nonnegative-valued, or all of the above, the problem of deciding its weak membership in either the spectral or nuclear norm unit ball in $\mathbb{C}^{n \times n \times n \times n}$ remains NP-hard. We provide a direct proof by showing that the clique number of a graph (well-known to be NP-hard) is the spectral norm of a $4$-tensor satisfying these properties, and applying \cite{FL1} to deduce the corresponding result for nuclear norm. Since we do not regard $d$-tensors as special cases of $(d+1)$-tensors, we provide a simple argument for extending such hardness results to higher order, giving us the required NP-hardness when $d \ge 3$ (for real tensors) and $d \ge 4$ (for complex tensors).

These hardness results may be stated in an alternative form, namely, the nuclear $p$-norm of $2$-tensors, the nuclear norm of $3$-tensors over $\mathbb{R}$, and the nuclear norm of $4$-tensors over $\mathbb{R}$ and $\mathbb{C}$, are all not polynomial-time approximable to arbitrary accuracy. We provide some simple polynomial-time computable approximation bounds for the spectral and nuclear norms in Section~\ref{sec:approx}.

\section{Hilbert--Schmidt, spectral, and nuclear norms for higher-order tensors}\label{sec:norms}

We let $\mathbb{F}$ denote either $\mathbb{R}$ or $\mathbb{C}$ throughout this article. A result stated for $\mathbb{F}$ holds true for both $\mathbb{R}$ and $\mathbb{C}$.  Let $\mathbb{F}^{n_1\times \dots \times n_d} \coloneqq \mathbb{F}^{n_1} \otimes \dots \otimes \mathbb{F}^{n_d}$ be the space of $d$-tensors of dimensions $n_1,\dots,n_d\in\mathbb{N}$. If desired, these may be viewed as $d$-dimensional hypermatrices $A = (a_{i_1 \cdots i_d})$ with entries $a_{i_1\cdots i_d}\in \mathbb{F}$.

The \textit{Hermitian inner product} of two $d$-tensors  $A,B \in \mathbb{C}^{n_1 \times \dots \times n_d}$ is given by
\begin{equation}\label{eq:hip}
\langle A, B\rangle = \sum_{i_1,\dots, i_d =1}^{n_1,\dots,n_d} a_{i_1 \cdots i_d} \overline{b_{i_1 \cdots i_d}}.
\end{equation}
When restricted to $ \mathbb{R}^{n_1 \times \dots \times n_d}$, \eqref{eq:hip} becomes the Euclidean inner product. This induces the  \textit{Hilbert--Schmidt norm} on  $\mathbb{F}^{n_1 \times \dots \times n_d}$, denoted by
\[
\| A \| = \sqrt{\langle A, A \rangle} =\left(\sum_{i_1,\dots, i_d =1}^{n_1,\dots,n_d} |a_{i_1 \cdots i_d} |^2 \right)^{\frac{1}{2}}.
\]
We adopt the convention that an unlabeled $\|\cdot\|$ will always denote the Hilbert--Schmidt norm. When $d=1$, this is the $l^2$-norm of a vector in $\mathbb{C}^n$ and when $d =2$, this is the Frobenius norm of a matrix in $\mathbb{C}^{m \times n}$.  As an $\mathbb{F}$-vector space, $\mathbb{F}^{n_1\times \dots \times n_d} \simeq \mathbb{F}^n$ where $n=\prod_{k=1}^d  n_k$, and the Hilbert--Schmidt norm on $\mathbb{F}^{n_1\times \dots \times n_d} $ equals the Euclidean
norm on $\mathbb{F}^n$.

Let $A\in\mathbb{F}^{n_1\times \dots\times n_d}$. We define its \textit{spectral norm}  by
\begin{align}
\|A\|_{\sigma,\mathbb{F}}&\coloneqq \sup \biggl\{ \frac{|\langle A, x_1\otimes \dots \otimes x_d\rangle |}{\| x_1 \| \cdots \| x_d \|} :  0 \ne x_k \in \mathbb{F}^{n_k} \biggr\}, \label{def1tenspnrm} \\
\intertext{and its \textit{nuclear norm} by}
\|A\|_{*,\mathbb{F}} &\coloneqq \inf\Bigl\{\sum_{i=1}^r \|x_{1,i}\|\cdots \|x_{d,i}\| : A =  \sum_{i=1}^r x_{1,i}\otimes \dots \otimes x_{d,i}, \; x_{k,i} \in \mathbb{F}^{n_k}, \; r \in \mathbb{N} \Bigr\}.\label{def1tennrm}
\end{align}
It is straightforward to show that these may also be expressed respectively as
\begin{align}
\|A\|_{\sigma,\mathbb{F}} &= \sup \bigl\{|\langle A, u_1 \otimes \dots \otimes u_d \rangle| :  \lVert u_k \rVert = 1 \bigr\}, \label{def2tenspnrm} \\
\|A\|_{*,\mathbb{F}} &= \inf\Bigl\{\sum_{i=1}^r \lvert \lambda_i\rvert : A =  \sum_{i=1}^r \lambda_i u_{1,i}\otimes \dots \otimes u_{d,i}, \; \lVert u_{k,i} \rVert = 1, \; r \in \mathbb{N}\Bigr\}. \label{def2tennrm}
\end{align}
The Hilbert--Schmidt norm is clearly independent of the choice of base field, i.e., $A \in \mathbb{R}^{n_1 \times \dots \times n_d} \subseteq  \mathbb{C}^{n_1 \times \dots \times n_d}$ has the same Hilbert--Schmidt norm whether it is regarded as a real tensor, $A \in \mathbb{R}^{n_1 \times \dots \times n_d}$, or a complex tensor, $A \in \mathbb{C}^{n_1 \times \dots \times n_d}$.  As we will see, this is not the case for spectral and nuclear norms when $d > 2$, which is why there is a subscript $\mathbb{F}$ in their notations.  
When $\mathbb{F} =\mathbb{C}$, the absolute value in \eqref{def1tenspnrm} and \eqref{def2tenspnrm} may replaced by the real part, giving
\[
\|A\|_{\sigma,\mathbb{C}}=\sup_{x_k \ne 0} \frac{\operatorname{Re}( \langle A, x_1 \otimes \dots \otimes x_d\rangle)}{\| x_1 \| \cdots \| x_d \|} =  \sup_{\lVert u_k \rVert = 1}  \operatorname{Re}( \langle A, u_1 \otimes \dots \otimes u_d \rangle).
\]
Henceforth we will adopt the convention that whenever the discussion holds for both $\mathbb{F} = \mathbb{R}$ and $\mathbb{C}$, we will drop the subscript $\mathbb{F}$ and write
\[
\|\cdot\|_{\sigma} = \|\cdot\|_{\sigma, \mathbb{F}} \qquad \text{and} \qquad \|\cdot\|_{*} = \|\cdot\|_{*,\mathbb{F}}.
\]

By \eqref{def1tenspnrm} and \eqref{def1tennrm}, we have
\[
\lvert \langle A, B \rangle \rvert \le \lVert A \rVert_{\sigma } \lVert B \rVert_{* }.
\]
In fact they are dual norms \cite[Lemma~21]{LC2}  since
\[
\lVert A \rVert_{*}^*  =   \sup_{\lVert B \rVert_{*} \le 1 } \lvert \langle A, B \rangle \rvert 
\le \sup_{\lVert B \rVert_{*} \le 1}\lVert A \rVert_{\sigma } \lVert B \rVert_{* }  =\lVert A \rVert_{\sigma},
\]
and on the other hand, it follows from $\lvert \langle A, B \rangle \rvert \le \lVert A \rVert_*^* \lVert B \rVert_*$ that
\[
\lVert A \rVert_{\sigma}  =   \sup_{ \lVert x_k \rVert  =1} \lvert \langle A, x_{1}\otimes\dots\otimes x_{d}\rangle \rvert 
\le \sup_{ \lVert x_k \rVert  =1} \lVert A \rVert_{*}^* \lVert x_{1}\otimes\dots\otimes x_{d}\rVert_{*} = \lVert A \rVert_{*}^*.
\]
It is also easy to see that
\[
\|x_1 \otimes \dots \otimes x_d \| = \|x_1 \otimes \dots \otimes x_d \|_{\sigma} = \|x_1 \otimes \dots \otimes x_d \|_{*} = \| x_1 \| \cdots \| x_d \|.
\]
In fact, the following generalization is clear from the definitions \eqref{def1tenspnrm} and \eqref{def1tennrm}.
\begin{proposition}\label{prop:mult}
Let $A \in \mathbb{F}^{n_1 \times \dots \times n_d}$ and $x_1 \in \mathbb{F}^{m_1}, \dots, x_e \in \mathbb{F}^{m_e}$. Then
\begin{align*}
\|A \otimes x_1 \otimes \dots \otimes x_e \|_{\sigma, \mathbb{F}} &=  \|A \|_{\sigma, \mathbb{F}}  \| x_1 \| \cdots \| x_e \|,\\
\|A \otimes x_1 \otimes \dots \otimes x_e \|_{*, \mathbb{F}} &=  \|A \|_{*, \mathbb{F}}  \| x_1 \| \cdots \| x_e \|.
\end{align*}
\end{proposition}

In this article, we undertake a coordinate dependent point-of-view for broader appeal --- a $d$-tensor is synonymous with a $d$-dimensional hypermatrix. Nevertheless we could also have taken a coordinate-free approach. A $d$-tensor is an element of a tensor product of $d$ vector spaces $V_1, \dots, V_d$ and choosing a basis on each of these vector spaces allows us to represent the $d$-tensor $\mathbf{A} \in V_1 \otimes \dots \otimes V_d$ as a $d$-hypermatrix $A \in \mathbb{F}^{n_1 \times \dots \times n_d}$. Strictly speaking, the $d$-hypermatrix $A$ is a coordinate representation of the $d$-tensor $\mathbf{A}$ with respect to our choice of bases; a difference choice of bases would yield a different hypermatrix for the same tensor \cite{Lim13}.

This can be extended to tensor product of $d$ norm spaces $(V_1, \|\cdot\|_1),\dots,(V_d, \|\cdot\|_d)$ or $d$ inner product spaces $(V_1,\langle \cdot, \cdot \rangle_1), \dots, (V_d,\langle \cdot, \cdot \rangle_d)$. For inner product spaces, defining an inner product on rank-one tensors by
\[
\langle u_1 \otimes \dots \otimes u_d, v_1 \otimes \dots \otimes v_d \rangle\coloneqq \langle u_1, v_1 \rangle_1 \cdots \langle u_d, v_d \rangle_d,
\] 
and extending bilinearly to the whole of $ V_1 \otimes \dots \otimes V_d$ defines an inner product on  $ V_1 \otimes \dots \otimes V_d$. For norm spaces, there are two natural ways of defining a norm on $ V_1 \otimes \dots \otimes V_d$. Let $V_1^*,\dots, V_d^*$ be the dual spaces\footnote{For norm space $(V, \|\cdot\|)$, dual space $V^* \coloneqq \{ \varphi:V \to \mathbb{F} \;\text{linear functional}\}$ has \textit{dual norm} $\|\varphi\|^{*} \coloneqq \sup_{\|v\|=1} |\varphi(v)|$.} of $V_1,\dots, V_d$. Then
\begin{align}
\|\mathbf{A}\|_{\sigma}&\coloneqq \sup \biggl\{ \frac{|\varphi_1 \otimes \dots \otimes \varphi_d (\mathbf{A}) |}{\| \varphi_1 \|_1^* \cdots \| \varphi_d \|_d^*} :  0 \ne \varphi_k \in V_k^* \biggr\},\label{eq:proj}\\
\|\mathbf{A}\|_{*} &\coloneqq \inf\Bigl\{\sum_{i=1}^r \| v_{1,i} \|_1 \cdots \| v_{d,i} \|_d : \mathbf{A} =  \sum_{i=1}^r v_{1,i}\otimes \dots \otimes v_{d,i}, \; v_{k,i} \in V_k, \; r \in \mathbb{N} \Bigr\},\label{eq:inj}
\end{align}
i.e., essentially the spectral and nuclear norm that we defined in \eqref{def1tenspnrm} and \eqref{def1tennrm}.

For the special case $d= 2$, \eqref{eq:proj} and \eqref{eq:inj} are the well-known \emph{injective} and \emph{projective norms} \cite{DF, Gro, PST, Ryan, Scha, Wong}. In operator theory, $V_1,\dots,V_d$ are usually infinite-dimensional Banach or Hilbert spaces and so one must allow $r = \infty$ in \eqref{eq:inj}. Also, the tensor product $\otimes$ has to be more carefully defined (differently for \eqref{eq:proj} and \eqref{eq:inj}) so that these norms are finite-valued on  $ V_1 \otimes  V_2$.

We are primarily interested in the higher-order case $d \ge 3$ in this article and all our spaces will be finite-dimensional to avoid such complications.

\section{Tensor nuclear norm is special}\label{sec:special}

We would like to highlight that \eqref{def1tennrm} is the definition of tensor nuclear norm as originally defined by Grothendieck \cite{Gro} and Schatten \cite{Scha}. An alternate definition of `tensor nuclear norm' as the average of nuclear norms of matrices obtained from flattenings of a tensor has gained recent popularity. While this alternate definition may be useful for various purposes, it is nevertheless not the definition commonly accepted in mathematics \cite{DF,Ryan,PST,Wong} (see also \cite{Fri83, LC2}). In particular, the nuclear norm defined in  \eqref{def1tennrm} is precisely the dual norm of the spectral norm in \eqref{def1tenspnrm}, is naturally related to the notion of tensor rank \cite{Lim13}, and has physical meaning --- for a $d$-Hermitian tensor $A \in(\mathbb{C}^{n_1 \times \dots \times n_d})^2$ representing a  density matrix,  $\|A\|_{*,\mathbb{C}} = 1$ if and only if $A$ is $d$-partite separable\footnote{This result appeared in an earlier preprint version of this article, see \url{https://arxiv.org/abs/1410.6072v1}, but has been moved to a more specialized article \cite{FL2} focusing on quantum information theory.} \cite{FL2}.
As such, a tensor nuclear norm in this article will always be the one in \eqref{def1tennrm} or its equivalent expression \eqref{def2tennrm}.

One might think that it is possible to extend  \eqref{def2tennrm} to get a definition of `Schatten $p$-norm' for any $p > 1$. Let us take $d = 3$ for illustration.  Suppose we define
\begin{equation}\label{eq:pnorm}
\nu_p(A) \coloneqq \inf\Bigl\{\left[  \sum_{i=1}%
^{r}\lvert\lambda_i\rvert^{p}\right]  ^{1/p} : A =\sum
_{i=1}^{r}\lambda_i u_i\otimes v_i\otimes w_i,\;
\lVert u_i\rVert=\lVert v_i\rVert=\lVert w_i\rVert%
=1,\;r\in\mathbb{N}\Bigr\}.
\end{equation}
Then $\nu_1 = \| \cdot \|_*$ but in fact $\nu_p$  is identically zero for all $p > 1$.
To see this, write $u \otimes v \otimes w$ as a sum of $2^n$ identical terms
\[
u \otimes v \otimes w = \tfrac{1}{2^n} u \otimes v \otimes w + \dots +  \tfrac{1}{2^n} u \otimes v \otimes w
\]
and observe that if $p>1$, then
\[
\inf_{n \in \mathbb{N}}\left[  \sum_{i=1}^{2^n} 2^{-np}\right]  ^{1/p}=\lim_{n\to\infty}2^{-n(p-1)/p}  = 0.
\]
This of course also applies to the case  $d=2$ but note that in this case we may impose orthonormality on the factors, i.e.,
\[
\nu_p( A)\coloneqq \inf\Bigl\{\left[  \sum_{i=1}%
^{r}\lvert\lambda_i\rvert^{p}\right]  ^{1/p} : A =\sum
_{i=1}^{r}\lambda_iu_i\otimes v_i, \; \langle u_i, u_j \rangle = \delta_{ij} = \langle v_i, v_j \rangle, \; r\in\mathbb{N}\Bigr\},
\]
and the result gives us precisely the matrix  Schatten $p$-norm.
This is not possible when $d > 2$. A $d$-tensor $A \in \mathbb{F}^{n_1 \times \dots \times n_d}$ is said to be \textit{orthogonally decomposable} \cite{Tong} if it has an \textit{orthogonal decomposition} given by
\[
A =  \sum_{i=1}^r \lambda_i u_{1,i}\otimes \dots \otimes u_{d,i}, \qquad \langle u_{k,i}, u_{k,j} \rangle =\delta_{ij}, \quad i, j = 1,\dots,n_k,\; k = 1,\dots,d.
\]
There is no loss of generality if we further assume that $\lambda_1 \ge \dots \ge \lambda_r > 0$.
An orthogonal decomposition does not exist when $d \ge 3$, as a simple dimension count would show. Nonetheless we would like to point out that this notion has been vastly generalized in \cite{Derk}.

The case $p =1$ is also special. In this case \eqref{eq:pnorm} reduces to  \eqref{def2tennrm} (for $d =3$), which indeed defines a norm for any $d$-tensors.
\begin{proposition}[Tensor nuclear norm]\label{prop:tnn}
The expression in \eqref{def1tennrm}, or equivalently \eqref{def2tennrm}, defines a norm on $\mathbb{F}^{n_1\times\cdots\times n_d}$. Furthermore, the infimum is attained and $\inf$ may be
replaced by $\min$ in \eqref{def1tennrm}.
\end{proposition}
\begin{proof}
Consider the set of all norm-one rank-one tensors, 
\[
\mathcal{E} \coloneqq \{u_1 \otimes \dots \otimes u_d \in \mathbb{F}^{n_1 \times \dots \times n_d} :  \|u_1 \| = \dots = \| u_d\|=1\}.
\]  
The Hilbert--Schmidt norm is strictly convex, i.e., for $A, B \in \mathbb{F}^{n_1\times\cdots\times n_d}$,
$\|A + B\|<2$ whenever  $A\ne B$, $\|A\|=\|B\|=1$. 
Hence in $\mathbb{F}^{n_1\times\cdots\times n_d}$ the extreme points of the unit ball are precisely the points on the unit sphere.
It follows that any rank-one tensor $A\in  \mathcal{E}$ is not a convex combination of any finite number of points in $\mathcal{E}\setminus \{A\}$.
Let $\mathcal{C}$ be the convex hull of $\mathcal{E}$.
Then $\mathcal{C}$ is a balanced convex set with $0$ as an interior point and so it must be a unit ball of some
norm $\nu$ on $\mathbb{F}^{n_1\times\cdots\times n_d}$. Clearly $\nu(A)=1$ for all
$A\in \mathcal{E}$.  So if
\[
A =\sum_{i=1}^r\lambda_i u_{1,i} \otimes \dots \otimes u_{d,i}, \quad \|u_{1,i} \otimes \dots \otimes u_{d,i}\|=1,
\]
then
\[
\sum_{i=1}^r |\lambda_i|\ge \nu(A).
\]
Hence $\|A\|_{*}\ge \nu(A)$.  We claim that $\|A\|_{*}= \nu(A)$.  Assume first that $\nu(A)=1$.
Then $A\in \{B \in \mathbb{F}^{n_1\times\cdots\times n_d} : \nu(B) = 1\} =  \mathcal{C}$.  So $A$ is a convex combination of a finite number of points in  $\mathcal{E}$, i.e.,
\[
A=\sum_{i=1}^r \lambda_i  u_{1,i}\otimes \dots \otimes u_{d,i}, \quad \| u_{1,i}\otimes \dots \otimes u_{d,i}\|=1, \; \lambda_1,\dots,\lambda_r>0, \; \sum_{i=1}^ r \lambda_i=1.
\]
By the definition of nuclear norm \eqref{def2tennrm}, $\|A\|_{*}\le 1=\nu(A)$.  So $\|A\|_{*}=1$ and the above decomposition of $A$ attains its nuclear norm. Thus if $\nu(A)=1$, the infimum in \eqref{def2tennrm} is attained.  For general $A\ne 0$, we consider $B=\frac{1}{\nu(A)} A$.
As $\|B\|_{*}=\nu(B)=1$, we have $\nu(A)=\|A\|_{*}$ and the infimum in \eqref{def2tennrm} is likewise attained.
\end{proof}

\section{Nuclear decompositions of tensors}\label{sec:decomp}

We will call the nuclear norm attaining decomposition in Proposition~\ref{prop:tnn} a nuclear decomposition for short, i.e., for $A \in \mathbb{F}^{n_1\times\cdots\times n_d}$,
\begin{equation}\label{defmindecT}
A=\sum_{i=1}^r x_{1,i} \otimes \dots \otimes x_{d,i}
\end{equation}
is a \textit{nuclear decomposition} over $\mathbb{F}$ if and only if 
\begin{equation}\label{eqmindecT}
\|A\|_{*,\mathbb{F}}=\sum_{i=1}^r  \|x_{1,i}\|\cdots  \|x_{d,i}\|,
\end{equation}
where $x_{k,i} \in \mathbb{F}^{n_k}$, $k =1,\dots,d$, $i =1,\dots,r$. We define the \textit{nuclear rank} of  $A \in \mathbb{F}^{n_1\times\cdots\times n_d}$ by
\begin{equation}\label{eq:nrank}
\operatorname{rank}_*(A) \coloneqq \min\Bigl\{r \in \mathbb{N} :  A=\sum_{i=1}^r x_{1,i} \otimes \dots \otimes x_{d,i}, \; \|A\|_{*,\mathbb{F}}=\sum_{i=1}^r  \|x_{1,i}\|\cdots  \|x_{d,i}\| \Bigr\},
\end{equation}
and we will call \eqref{defmindecT} a \textit{nuclear rank decomposition} if $r =\operatorname{rank}_*(A)$.   Alternatively, we may write the decomposition in a form that resembles the matrix \textsc{svd}, i.e.,
\begin{equation}\label{svdnd}
A =\sum_{i=1}^r\lambda_i u_{1,i} \otimes \dots \otimes u_{d,i}
\end{equation}
is a nuclear decomposition over $\mathbb{F}$ if and only if
\[
\|A\|_{*,\mathbb{F}}=\sum_{i=1}^r  \lambda_i \qquad \text{and} \qquad \lambda_1 \ge \dots \ge \lambda_r > 0, \quad \| u_{k,i} \| = 1,
\]
where $u_{k,i} \in \mathbb{F}^{n_k}$, $k=1,\dots,d$, $i = 1,\dots, r$. Unlike the matrix \textsc{svd}, $\{u_{k,1},\dots, u_{k,r}\}$ does not need to be orthonormal.

The following lemma provides a way that allows us to check, in principle, when a given decomposition is a nuclear decomposition.
\begin{lemma}\label{necsufmindec}  Let $A\in \mathbb{F}^{n_1 \times \dots \times n_d}$.  Then \eqref{defmindecT} is a nuclear decomposition over  $\mathbb{F}$ if and only if
there exists $0\ne B\in  \mathbb{F}^{n_1 \times \dots \times n_d}$ with
\begin{equation}\label{necsufmindec1}
\langle B, x_{1,i} \otimes \dots \otimes x_{d,i} \rangle =\|B\|_{\sigma,\mathbb{F}}\|x_{1,i}\|\cdots  \|x_{d,i}\| , \qquad  i = 1,\dots,r.
\end{equation}
Alternatively,  \eqref{svdnd} is a nuclear decomposition over  $\mathbb{F}$ if and only if there exists $0 \ne B\in  \mathbb{F}^{n_1 \times \dots \times n_d}$ with
\[
\langle B, u_{1,i} \otimes \dots \otimes u_{d,i} \rangle = \|B\|_{\sigma,\mathbb{F}}, \qquad  i = 1,\dots,r.
\]%
\end{lemma}
\begin{proof}
Since the nuclear and spectral norms are dual norms, $\operatorname{Re} \langle A,B \rangle \le \|A\|_{*,\mathbb{F}}\|B\|_{\sigma,\mathbb{F}}$.  Suppose $\|B\|_{\sigma,\mathbb{F}}=1$ and $A\ne 0$.
Then $\operatorname{Re} \langle A,B \rangle = \|A\|_{*,\mathbb{F}}\|B\|_{\sigma,\mathbb{F}}$ if and only if the real functional $X \mapsto \operatorname{Re} \langle X,B \rangle $ is a supporting hyperplane of the ball
$\{X\in  \mathbb{F}^{n_1 \times \dots \times n_d} : \|X\|_{*,\mathbb{F}}\le \|A\|_{*,\mathbb{F}}\}$ at the point $X=A$. So $\operatorname{Re} \langle A,B \rangle = \|A\|_{*,\mathbb{F}}$ is always attained for some $B$ with  $\|B\|_{\sigma,\mathbb{F}}=1$.

Suppose \eqref{defmindecT} is a nuclear decomposition, i.e., \eqref{eqmindecT}  holds. Let $B\in \mathbb{F}^{n_1 \times \dots \times n_d}$, $\|B\|_{\sigma,\mathbb{F}}=1$ be
such that $\operatorname{Re}\langle A,B\rangle=\|A\|_{*,\mathbb{F}}$. Then
\[
\|A\|_{*,\mathbb{F}}= \operatorname{Re} \langle A,B \rangle  =\sum_{i=1}^r \operatorname{Re} \langle x_{1,i} \otimes \dots \otimes x_{d,i}, B \rangle \le\sum_{i=1}^r \prod_{k=1}^d \|x_{k,i}\|=\|A\|_{*,\mathbb{F}}.
\]
Therefore equality holds and we have \eqref{necsufmindec1}.

Suppose \eqref{necsufmindec1}  holds.  We may assume without loss of generality that $\|B\|_{\sigma,\mathbb{F}}=1$ and $\prod_{k=1}^d \|x_{k,i}\|>0$ for each $i=1,\dots,r$.
Then
\[
\|A\|_{*,\mathbb{F}}=\|A\|_{*,\mathbb{F}}\|B\|_{\sigma,\mathbb{F}}\ge \operatorname{Re} \langle A,B \rangle 
=\sum_{i=1}^r \langle x_{1,i} \otimes \dots \otimes x_{d,i}, B \rangle
=\sum_{i=1}^r \prod_{k=1}^d \|x_{k,i}\|.
\]
It follows from the minimality in \eqref{def1tennrm} that \eqref{defmindecT} is a nuclear decomposition of $A$.
\end{proof}

As an illustration of Lemma~\ref{necsufmindec}, we prove that for an orthogonally decomposable tensor, every orthogonal decomposition is a nuclear decomposition, a special case of \cite[Theorem~1.11]{Derk}.
\begin{corollary}
Let $A \in \mathbb{F}^{n_1 \times \dots \times n_d}$ be orthogonally decomposable and
\begin{equation}\label{eq:odeco}
A =  \sum_{i=1}^r \lambda_i u_{1,i}\otimes \dots \otimes u_{d,i}, \quad \langle u_{k,i}, u_{k,j} \rangle =\delta_{ij},
\end{equation}
be an orthogonal decomposition. Then
\[
\|A\| =\Bigl( \sum_{i=1}^r |\lambda_i|^2 \Bigr)^{1/2}, \qquad \|A\|_{\sigma,\mathbb{F}} = \max_{i=1,\dots,r} |\lambda_i |, \qquad \|A\|_{*,\mathbb{F}} = |\lambda_1 | + \dots +| \lambda_r|.
\]
\end{corollary}
\begin{proof}
The expression for Hilbert--Schmidt norm is immediate from Pythagoras theorem since $\{ u_{1,i}\otimes \dots \otimes u_{d,i}: i=1,\dots,r\}$ is orthonormal. We may assume that  $\lambda_1 \ge \dots \ge \lambda_r > 0$. Let $v_k \in \mathbb{F}^{n_k}$, $k=1,\dots,d$, be unit vectors.  Clearly, $\lvert\langle u_{k,i}, v_k\rangle\rvert\le 1$ for all $i$ and $k$. By Bessel's inequality,
$\sum_{i=1}^r  \lvert \langle u_{k,i},v_k\rangle\rvert^2\le \lvert v_k\rvert^2=1$ for $k=1,2$.  Hence
\begin{align*}
|\langle A, v_1 \otimes \dots \otimes v_d\rangle| &\le \sum_{i=1}^r  \lambda_i | \langle u_{1,i}\otimes \dots \otimes u_{d,i}, v_1 \otimes \dots \otimes v_d \rangle| \\
&= \sum_{i=1}^r  \lambda_i\prod_{k=1}^d |\langle u_{k,i},v_k \rangle| \le \lambda_1\sum_{i=1}^r  |\langle u_{1,i},v_1\rangle|| \langle u_{2,i},v_2\rangle|\\
&\le \lambda_1 \Bigl(\sum_{i=1}^r  |\langle u_{1,i},v_1\rangle |^2\Bigr)^{1/2} \Bigl(\sum_{i=1}^r  |\langle u_{2,i},v_2\rangle |^2\Bigr)^{1/2}\le \lambda_1.
\end{align*}
Choose $v_k=u_{k,i}$ for $k=1,\dots,d$ to deduce that $\|A \|_{\sigma,\mathbb{F}} =\lambda_1 = \max_{i=1,\dots,r} \lambda_i $.  Now take $B\coloneqq \sum_{i=1}^r  u_{1,i}\otimes \dots \otimes u_{d,i}$ and observe that $\|B\|_{\sigma,\mathbb{F}} = 1$ and that $\langle B, u_{1,i}\otimes \dots \otimes u_{d,i}\rangle =1$ for all $i=1,\dots,r$.
Hence by Lemma~\ref{necsufmindec}, \eqref{eq:odeco} is a nuclear decomposition and $\|A\|_* = \sum_{i=1}^r \lambda_i$.
\end{proof}%

For $\mathbb{F} = \mathbb{R}$, we establish a generalization of nuclear decomposition that holds true for any finite-dimensional norm space $V$.  The next result essentially says that  `every norm is a nuclear norm' in an appropriate sense.
\begin{proposition}\label{prop:convcombextpts}
Let $V$ be a real vector space of dimension $n$ and $\nu : V \to [0,\infty)$ be a norm. Let $\mathcal{E}$ be the set of the extreme points of the unit ball $B_{\nu}:=\{x\in
V : \nu(x)\le 1\}$. If $\nu(x)=1$, then there exists a decomposition
\begin{equation}\label{gnd}
x=\sum_{i=1}^r \lambda_i x_i,
\end{equation}
where $\lambda_1,\dots, \lambda_r >0$, $\lambda_1 + \dots +\lambda_r=1$, and $x_1,\dots,x_r\in \mathcal{E}$ are linearly independent. Furthermore, for any $x\in V$,
\begin{equation}\label{charnux}
\nu(x)=\min\Bigl\{ \sum_{i=1}^n |\lambda_i | : x = \sum_{i=1}^n \lambda_i x_i,\; x_1,\dots,x_n \in \mathcal{E} \; \text{linearly independent}\Bigr\}.
\end{equation}
\end{proposition}  
\begin{proof}
Let $\nu(x)=1$.  By Krein--Milman, $x$ is a convex combination of the extreme points of $B_{\nu}$,
\[
x=\sum_{i=1}^r \lambda_i x_i, \quad x_1,\dots,x_r\in\mathcal{E},  \quad \lambda_1,\dots,\lambda_r > 0,\quad  \sum_{i=1}^r \lambda_i=1.
\] 
Let $r$ be minimum.  We claim that for such a minimum decomposition $x_1,\dots,x_r$ must be linearly independent.
Suppose not, then there is a non-trivial linear combination
\begin{equation}\label{eq:betas}
\sum_{i=1}^r \beta_i x_i=0.
\end{equation}
We claim that $\sum_{i=1}^r \beta_i=0$.  Suppose not. Then we may assume that $\sum_{i=1}^r \beta_i>0$ (if not, we  replace $\beta_i$ by $-\beta_i$ in \eqref{eq:betas}).
Choose $t>0$ such that  $\lambda_i-t\beta_i\ge 0$ for $i=1,\dots,r$.  Then
\[
1=\nu(x)=\nu\Bigl(\sum_{i=1}^r (\lambda_i-t\beta_i)x_i\Bigr)\le \sum_{i=1}^r (\lambda_i-t\beta_i)\nu(x_i)=\sum_{i=1}^r \lambda_i-t\beta_i=1-t\sum_{i=1}^r\beta_i<1,
\]
a contradiction. Hence $\sum_{i=1}^r \beta_i=0$. By our earlier assumption that the linear combination in \eqref{eq:betas} is nontrivial, not all $\beta_i$'s are zero; so we may choose $t > 0$ such that $\lambda_i-t\beta_i\ge 0$ for all $i=1,\dots,r$  and $\lambda_i-t\beta_i=0$ for at least one $i$. In which case the decomposition $x=\sum_{i=1}^r (\lambda_i-t\beta_i)x_i$ contains fewer than $r$ terms, contradicting the minimality of $r$.  Hence $x_1,\dots,x_r$ are linearly independent. Clearly $r \le n$.

We now prove the second part.  Since $-B_{\nu}=B_{\nu}$, it follows that $-\mathcal{E}=\mathcal{E}$.  
Since $B_{\nu}$ has nonempty interior,  $\operatorname{span}_\mathbb{R}(\mathcal{E}) = V$.  So any $x \in V$ may be written as a linear combination
\begin{equation}\label{decxM}
x = \sum_{i=1}^n \lambda_i x_i, \qquad x_1,\dots,x_n \in \mathcal{E} \; \text{linearly independent}.
\end{equation}
Since $\nu(x_i)=1$ for $i=1,\dots,n$, $\nu(x)\le \sum_{i=1}^n |\lambda_i|$, and thus the right-hand side of \eqref{charnux} is not less than $\nu(x)$.  It remains to show that there exist linearly independent $x_1,\dots,x_n\in \mathcal{E}$ such that the the decomposition \eqref{decxM} attains $\nu(x)=\sum_{i=1}^n |\lambda_i|$.  This is trivial for $x=0$  and we may assume that $x\ne 0$. Upon normalizing, we may further
assume that $\nu(x)=1$. By the earlier part, we have a convex decomposition $x=\sum_{i=1}^r \lambda_i x_i$ where $x_1,\dots,x_r\in\mathcal{E}$ and $\sum_{i=1}^r \lambda_i=1$. If $r=n$, we are done. If $r < n$, we extend $x_1,\dots,x_r$ to $x_1,\dots,x_n\in\mathcal{E}$, a basis of $V$; note that this is always possible since $\mathcal{E}$ is a spanning set.  Then $x=\sum_{i=1}^n \lambda_i x_i$ by  setting $\lambda_i \coloneqq 0$ for $i = r+1,\dots,  n$. Hence $1=\nu(x)=\sum_{i=1}^n |\lambda_i|$.
\end{proof}

For any $0 \ne x \in V$, we may apply Proposition~\ref{prop:convcombextpts} to the unit vector $x/\nu(x)$ to obtain a \textit{nuclear  decomposition} for $x$,
\begin{equation}\label{gnd1}
x=\lambda_1 x_1+ \dots +\lambda_r x_r, \qquad  \nu(x) = \lambda_1+ \dots + \lambda_r, \qquad \lambda_1 \ge \dots \ge \lambda_r >0,
\end{equation}
where $x_1,\dots,x_r$ are extreme points of $B_\nu$.
We define  \textit{nuclear rank} of $x \in V$, denoted by $\operatorname{rank}_\nu ( x)$, to be the minimum $r \in \mathbb{N}$ such that \eqref{gnd1} holds. We set $\operatorname{rank}_\nu ( x) = 0$ iff $x = 0$. A nuclear decomposition \eqref{gnd1} where $r =\operatorname{rank}_\nu ( x) $ is called a \textit{nuclear rank decomposition}.  Note that the linear independence of $x_1,\dots,x_r$ in \eqref{gnd1} is automatic if it is  a nuclear rank decomposition.
\begin{proposition}\label{prop:usc}  
Let $V$ be a real vector space of dimension $n$ and $\nu : V \to [0,\infty)$ be a norm. Suppose $\mathcal{E}$, the set of the extreme points of the unit ball $B_{\nu}$, is compact. 
Then the nuclear rank $\operatorname{rank}_\nu : V \to \mathbb{R}$ is a upper semicontinuous function, i.e., if $(x_m)_{m=1}^\infty$  is a convergent sequence in $V$ with $\operatorname{rank}_\nu  (x_m)  \le r$ for all $m \in \mathbb{N}$, then $x = \lim_{m\to \infty} x_m $ must  have  $\operatorname{rank}_\nu  (x) \le r$.
\end{proposition}
\begin{proof}
 For each $m \in \mathbb{N}$, since $\operatorname{rank}_\nu  (x_m)  \le r$, $x_m$ has a nuclear decomposition $x_m=\sum_{i=1}^{r}\lambda_{m,i} x_{m,i}$ with $\sum_{i=1}^r \lambda_{m,i}= \nu(x_m)$, $\lambda_{m,1},\dots,\lambda_{m,r} \ge 0$, and $x_{m,1},\dots, x_{m,r} \in \mathcal{E}$. Since $\mathcal{E}$ is compact, by passing through subsequences $r$ times, we obtain a nuclear decomposition $x = \sum_{i=1}^r \lambda_i x_i$ with $\sum_{i=1}^r \lambda_i= \nu(x)$,  $\lambda_1,\dots,\lambda_r \ge 0$, and $x_1,\dots,x_r \in\mathcal{E}$.  Hence $\operatorname{rank}_\nu  (x) \le r$.
\end{proof}
If $V = \mathbb{R}^{n_1 \times \dots \times n_d}$ and $\nu = \|\cdot\|_{*,\mathbb{R}}$, then $\mathcal{E}=\{u_1\otimes\dots\otimes u_d : \|u_1\|=\dots=\|u_d\|=1\}$ and \eqref{gnd1} gives a nuclear decomposition in the sense it was defined in \eqref{eq:nrank}. Also, since $\mathcal{E}$ is compact, tensor nuclear rank  is upper semicontinuous. The lack of upper semicontinuity in tensor rank has been a source of many problems \cite{DSL}, particularly the best rank-$r$ approximation problem for $d$-tensors does not have a solution when $r \ge 2$ and $d \ge 3$. We note that the use of nuclear rank would  alleviate this problem.
\begin{corollary}
For any $A \in \mathbb{R}^{n_1 \times \dots \times n_d}$, the best nuclear rank-$r$ approximation problem
\[
\operatorname{argmin}\{ \| A - X \| : \operatorname{rank}_*(X) \le r \}
\]
always has a solution.
\end{corollary}
\begin{proof}
By Proposition~\ref{prop:usc}, $\mathcal{S} = \{X \in \mathbb{R}^{n_1 \times \dots \times n_d}:  \operatorname{rank}_*(X) \le r \}$ is a closed set and the result follows from the fact that in any metric space the distance between a point $A$ and a closed set $\mathcal{S}$ must be attained by some $X \in \mathcal{S}$.
\end{proof}

%
%

\section{Analogue of Comon's conjecture and Banach's theorem for nuclear norm}\label{sec:banach}

We write $\mathsf{T}^d(\mathbb{F}^n) \coloneqq  (\mathbb{F}^n)^{\otimes d} = \mathbb{F}^{n \times \dots \times n}$ for the space of cubical $d$-tensors and  $\mathsf{S}^d(\mathbb{F}^n)$ for the subspace of symmetric $d$-tensors in $\mathsf{T}^d(\mathbb{F}^n)$.  See \cite{CGLM} for definition and basic properties of symmetric tensors.
Let $A\in \mathsf{S}^{d}(\mathbb{F}^{n})$. Comon's conjecture \cite{CGLM} asserts that the rank and symmetric rank of a symmetric tensor are always equal,
\[
\min \left\{ r : A = \sum_{i=1}^{r}\lambda_{i}u_{1,i}\otimes \dots \otimes u_{d,i}\right\} \overset{?}{=} 
\min \left\{ r : A = \sum_{i=1}^{r}\lambda_{i}u_{i}^{\otimes d}\right\}.
\]
Banach's theorem \cite{Ban38,Fri13} on the other hand shows that the analogous assertion for spectral norm is true over both $\mathbb{R}$ and $\mathbb{C}$,
\begin{equation}\label{banachthm}
\sup_{x_1,\dots,x_d\ne0}\frac{\lvert \langle A, x_1 \otimes \dots \otimes x_d)\rvert}{\lVert x_1\rVert \cdots \lVert x_d\rVert} =
\sup_{x\ne0}\frac{\lvert \langle A, x^{\otimes d} \rangle \rvert}{\lVert x\rVert^d} .
\end{equation}
Here we show that the analogous assertion for nuclear norm is also true over both $\mathbb{R}$ and $\mathbb{C}$,
\begin{equation}\label{banachnuc1}
\inf\left\{  \sum_{i=1}^{r}\lvert\lambda_{i} \rvert :A=\sum_{i=1}^{r}\lambda_{i}u_{1,i}\otimes \dots \otimes u_{d,i}\right\}
= \inf\left\{  \sum_{i=1}^{r}\lvert\lambda_{i} \rvert :A=\sum_{i=1}^{r}\lambda_{i}u_{i}^{\otimes d}\right\}.
\end{equation}
We will first prove a slight variation of \eqref{banachnuc1} over $\mathbb{R}$ below.  Note that \eqref{banachnuc1} follows from \eqref{banachnuc}. If $d$ is odd in \eqref{banachnuc1}, we may drop the $\varepsilon_i$'s. 
\begin{theorem}\label{thm:bnuclear}
Let $A\in\mathsf{S}^d(\mathbb{R}^n)$. Then
\begin{equation}\label{banachnuc}
\|A\|_{*,\mathbb{R}}=\min\left\{\sum_{i=1}^r \|x_i\|^d : A=\sum_{i=1}^r \varepsilon_i x^{\otimes d}_i, \; \varepsilon_i\in\{-1,1\}\right\}.
\end{equation}
The infimum is taken over all  possible symmetric rank-one decompositions of $A$ with $r \in \mathbb{N}$ and is attained (therefore denoted by minimum).
\end{theorem}
\begin{proof}
Let $\mathcal{C}\coloneqq \operatorname{conv}(\mathcal{E})\subseteq \mathsf{T}^d(\mathbb{R}^n)$ be the convex hull of all vectors of the form
\[
\mathcal{E}\coloneqq \{\pm x^{\otimes d} : x\in \mathbb{R}^n,\; \|x\|=1\}.
\]
As $x^{\otimes d} +(-x^{\otimes d})=0$, $\mathcal{C}$ is a symmetric set in $\mathsf{S}^d(\mathbb{R}^n)$.  Since any symmetric tensor is a linear combination of symmetric rank-one terms $x^{\otimes d}$, $\mathcal{C}$ has nonempty interior in $\mathsf{S}^d(\mathbb{R}^n)$.  Hence $\mathcal{C}$ is the unit ball of some norm $\nu:\mathsf{S}^d(\mathbb{R}^n)\to [0,\infty)$.   Note that $\nu(x^{\otimes d})\le 1$
for $\|x\|= 1$.  We claim that each point of $\mathcal{E}$ is an extreme point of $\mathcal{C}$.  Indeed, consider the unit ball of the Hilbert--Schmidt norm $\{ A \in \mathsf{S}^d(\mathbb{R}^n) : \|A\|\le 1\}$. Note that $\lVert\pm x^{\otimes d}\rVert=1$ for $\|x\|=1$, and as $\|\cdot\|$ is a strictly convex function, no point on $\mathcal{E}$ is a convex combination of other points of $\mathcal{E}$.
Hence $\nu(\pm x^{\otimes d})=\|x\|^d$ for $\|x\|=1$.  The homogeneity of $\nu$ implies that $\nu(\pm x^{\otimes d})=\|x\|^d$.

Suppose $A=\sum_{i=1}^r \alpha_ix^{\otimes d}_i$.  Then the triangle inequality for $\nu$ and the above equality yields $\nu(A)\le \sum_{i=1}^r |\alpha_i|\|x_i\|^d$.
By scaling the norm of $x_i$ appropriately, we may assume without loss of generality that $\alpha_i\in\{-1,1\}$ for $i = 1,\dots,r$.  Hence
\[
\nu(A)\le \inf\left\{\sum_{i=1}^r \|x_i\|^d : A=\sum_{i=1}^r \varepsilon_i x^{\otimes d}_i, \; \varepsilon_i\in\{-1,1\}\right\}.
\] 
We claim that the infimum is attained.  It is enough to consider the case $\nu(A)=1$. So $A\in \mathcal{C}$ and $A$ is a convex combination of the extreme points of $\mathcal{C}$, i.e., 
\begin{equation}\label{basineq}
A=\sum_{i=1}^r t_i\varepsilon_ix^{\otimes d}_i,\qquad\sum_{i=1}^r t_i=1,
\end{equation}
where  $t_i \ge 0$, $\|x_i\|=1$, $\varepsilon_i=\pm 1$, for all $i=1,\dots,r$.  Since $\dim_\mathbb{R} \mathsf{S}^d(\mathbb{R}^n)= \binom{n+d-1}{d}$, 
Caratheodory's theorem implies that $r\le 1+ \binom{n+d-1}{d}$.  The triangle inequality gives
\begin{equation}\label{basineq1}
1=\nu(A)\le \sum_{i=1}^r t_i\nu(x_i^{\otimes d})= \sum_{i=1}^r t_i=1.
\end{equation}
We deduce from \eqref{basineq} and \eqref{basineq1} that  $\nu(A)$ is given by the right-hand side of \eqref{banachnuc}.

Let $\nu^*$ be the dual norm of $\nu$ in $\mathsf{S}^d(\mathbb{R}^n)$.  By definition, 
\[
\nu^*(A)=\max_{B\in \mathsf{S}^d(\mathbb{R}^n),\; \nu(B)\le 1} \langle A,B \rangle =\max_{B\in \mathcal{E}} \langle A,B \rangle=\max_{\|x\|=1} | \langle A,x^{\otimes d} \rangle|.
\]
Since Banach's theorem \eqref{banachthm} may be written in the form $\|A\|_{\sigma,\mathbb{R}}=\max_{\|x\|=1} | \langle A,x^{\otimes d} \rangle|$, we get
\begin{equation}\label{eq:nustar}
\nu^*(A)=\|A\|_{\sigma,\mathbb{R}}.
\end{equation}
From the definition of nuclear norm \eqref{def1tennrm} and the fact that $\nu(A)$ is given by the right-hand of  \eqref{banachnuc}
we deduce that $\|A\|_{*,\mathbb{R}}\le \nu(A)$ for all $A\in\mathsf{S}^d(\mathbb{R}^n)$.

Let $\nu_1 : \mathsf{S}^d(\mathbb{R}^n)\to [0,\infty)$ be the nuclear norm  $\| \cdot \|_{*,\mathbb{R}}$ on
$\mathsf{T}^d(\mathbb{R}^n)$ restricted to  $\mathsf{S}^d(\mathbb{R}^n)$.  So $\nu_1(A)=\|A\|_{*,\mathbb{R}}$ for $A\in\mathsf{S}^d(\mathbb{R}^n)$.  We claim that $\nu=\nu_1$.  Suppose not.
Then the $\nu_1$ unit ball $\mathcal{C}_1\coloneqq \{A : \nu_1(A)\le 1\}$ must strictly contain the $\nu$ unit ball, i.e.,
$\mathcal{C} \subsetneq \mathcal{C}_1$.
Let $\nu_1^* : \mathsf{S}^d(\mathbb{R}^n)\to [0,\infty)$ be the dual norm of $\nu_1$. Let $\mathcal{C}^*$ and $\mathcal{C}_1^*$
be the unit balls of $\nu^*$ and $\nu_1^*$ respectively.  Then $\mathcal{C} \subsetneq \mathcal{C}_1$ implies that $\mathcal{C}_1^*\subsetneq \mathcal{C}^*$.
So there exists  $A\in\mathsf{S}^d(\mathbb{R}^n)$ such that $\nu_1^*(A) >\nu^*(A)$.  Hence
\[
\nu^*(A) < \nu_1^*(A)=\max_{B\in\mathsf{S}^d(\mathbb{R}^n),\; \|B\|_{*,\mathbb{R}}\le 1} \langle A,B \rangle\le \max_{B\in \mathsf{T}^d(\mathbb{R}^n),\; \|B\|_{*,\mathbb{R}}\le 1} \langle A,B \rangle=\|A\|_{\sigma,\mathbb{R}},
\]
which contradicts \eqref{eq:nustar}.
\end{proof}

The complex case may be deduced from the real case as follows. Note that the $\varepsilon_i$'s in \eqref{banachnuc} are unnecessary regardless of the order $d$ since $\mathbb{C}$ contains all $d$th roots of unity.
\begin{corollary}\label{cor:bnuclear}
Let $A\in\mathsf{S}^d(\mathbb{C}^n)$. Then
\[
\|A\|_{*,\mathbb{C}}=\min\left\{\sum_{j=1}^r \|x_j\|^d : A=\sum_{j=1}^r x^{\otimes d}_j\right\}.
\]
The infimum is taken over all  possible symmetric rank-one decompositions of $A$ with $r \in \mathbb{N}$ and is attained (therefore denoted by minimum).
\end{corollary}
\begin{proof}
We identify $\mathsf{T}^d(\mathbb{C}^n)$ with $\mathsf{T}^d(\mathbb{R}^n)\times \mathsf{T}^d(\mathbb{R}^n)$, i.e., we write $B\in \mathsf{T}^d(\mathbb{C}^n)$ as $B=X+iY$ 
where $X,Y\in \mathsf{T}^d(\mathbb{R}^n)$ and identify $B$ with $(X,Y)$.  On $\mathsf{T}^d(\mathbb{R}^n) \times \mathsf{T}^d(\mathbb{R}^n)$, we define a real inner product 
\[
\langle (X,Y),(W,Z) \rangle=\langle X,W\rangle +\langle Y,Z \rangle =\operatorname{Re}\langle X+iY ,W+iZ\rangle,
\]
under which the Hilbert--Schmidt norm on $\mathsf{T}^d(\mathbb{C}^n)$ is the same as the Hilbert--Schmidt norm on  $\mathsf{T}^d(\mathbb{R}^n) \times \mathsf{T}^d(\mathbb{R}^n)$.
The spectral norm on $\mathsf{T}^d(\mathbb{C}^n)$ defined in \eqref{def1tenspnrm} translates to a spectral norm on the real space $\mathsf{T}^d(\mathbb{R}^n) \times \mathsf{T}^d(\mathbb{R}^n)$.  Furthermore its dual norm on $\mathsf{T}^d(\mathbb{R}^n) \times \mathsf{T}^d(\mathbb{R}^n)$ is precisely the nuclear norm on $\mathsf{T}^d(\mathbb{C}^n)$ as defined in \eqref{def1tennrm}.
This follows from the observation that the extreme points of the nuclear norm unit ball in $\mathsf{T}^d(\mathbb{C}^n)$ is exactly
\[
\mathcal{E}=\{x_{1} \otimes \dots \otimes x_{d} : x_1,\dots,x_d\in \mathbb{C}^{n}, \; \lVert x_1 \rVert = \dots = \lVert x_d \rVert =1\}.
\] 
So $\mathsf{S}^d(\mathbb{C}^n)$ may be viewed as a real subspace of $\mathsf{S}^d(\mathbb{R}^n)\times \mathsf{S}^d(\mathbb{R}^n)$. We may repeat the arguments as in the real case and use Banach's theorem  \eqref{banachthm} for complex-valued symmetric tensors.
\end{proof}

An immediate consequence of Theorem~\ref{thm:bnuclear} and Corollary~\ref{cor:bnuclear}  is the existence of a \textit{symmetric nuclear decomposition} for symmetric $d$-tensors.
\begin{corollary}[Symmetric nuclear decomposition]\label{cor:snd}
Let $A\in\mathsf{S}^d(\mathbb{F}^n)$. Then there exists a decomposition
\[
A = \sum_{i=1}^{r}\lambda_{i}u_{i}^{\otimes d}
\]
with finite $r \in \mathbb{N}$,$r\le 1+\binom{n+d-1}{d}$, and $\lVert u_1 \rVert = \dots =\lVert u_r \rVert = 1$ such that
\[
\|A\|_{*,\mathbb{F}} =\lvert\lambda_1 \rvert + \dots + \lvert\lambda_r \rvert.
\]
\end{corollary}

As in \cite{Fri13} we may extend Theorem~\ref{thm:bnuclear} and Corollary~\ref{cor:bnuclear} to partially symmetric tensors.  Let $d_1,\dots,d_m \in \mathbb{N}$ and $d = d_1 + \dots + d_m$.
A $d$-tensor $A \in \mathsf{S}^{d_1}(\mathbb{F}^{n_1}) \otimes \dots \otimes \mathsf{S}^{d_m}(\mathbb{F}^{n_m})$ is called a $(d_1,\dots,d_m)$-symmetric  tensor.
The following analogue of Banach's theorem \eqref{banachthm} for such tensors was established in \cite{Fri13}:
\[
\|A\|_{\sigma,\mathbb{F}}=\max_{\| x_i\| = 1} |\langle A,x_1^{\otimes d_1} \otimes \dots \otimes x_m^{\otimes d_m}\rangle |
\]
for all $A\in\mathsf{S}^{d_1}(\mathbb{F}^{n_1}) \otimes \dots \otimes \mathsf{S}^{d_m}(\mathbb{F}^{n_m})$. Using this and the same arguments used to establish Theorem~\ref{thm:bnuclear} and Corollary~\ref{cor:bnuclear}, we may obtain the following. Note that the $\varepsilon_i$'s in \eqref{banachnucps} may be dropped in all cases except when $\mathbb{F}=\mathbb{R}$ and $d_1,\dots,d_m$ are all even integers.
\begin{corollary}
Let $A\in\mathsf{S}^{d_1}(\mathbb{F}^{n_1}) \otimes \dots \otimes \mathsf{S}^{d_m}(\mathbb{F}^{n_m})$. Then
\begin{equation}\label{banachnucps}
\|A\|_{*,\mathbb{F}}=\min \biggl\{\sum_{i=1}^r \|x_{1,i}\|^{d_1} \cdots \|x_{m,i}\|^{d_m} :
A=\sum_{i=1}^r \varepsilon_j x_{1,i}^{\otimes d_1} \otimes \dots \otimes x_{m,i}^{\otimes d_m},\; \varepsilon_i\in\{-1,1\}\biggr\}.
\end{equation}
\end{corollary}

\section{Base field dependence}\label{sec:base}

It is well-known \cite{Bry, DSL} that tensor rank is dependent on the choice of base fields when the order of the tensor $d \ge 3$. Take any  linearly independent  $x, y\in\mathbb{R}^{n}$ and let
$ z= x+i y \in\mathbb{C}^{n}$. If we define
\[
A \coloneqq x\otimes x\otimes x- x\otimes y%
\otimes y+ y\otimes x\otimes y+ y%
\otimes y\otimes x
=\frac{1}{2}(z\otimes\bar{z}\otimes\bar{z}%
+\bar{z}\otimes z\otimes z),
\]
then $\operatorname{rank}_{\mathbb{C}}(A) =2 < 3 =\operatorname{rank}_{\mathbb{R}}(A)$. We  show that the same is true for spectral and nuclear norms of $d$ tensors when $d \ge 3$.

\begin{lemma}\label{3tenexamnorm2}
Let $e_1, e_2 \in \mathbb{R}^2$ be the standard basis vectors. Define $B \in \mathbb{R}^{2 \times 2 \times 2} \subseteq \mathbb{C}^{2 \times 2 \times 2}$ by
\begin{align}\label{exam2}
B=\frac{1}{2}(e_1\otimes  e_1\otimes  e_2+e_1\otimes  e_2\otimes  e_1+e_2\otimes  e_1\otimes  e_1-e_2\otimes  e_2\otimes  e_2).
\end{align}
Then \eqref{exam2} is a nuclear decomposition over $\mathbb{R}$,  and
\[
\|B\|_{\sigma,\mathbb{R}}=\frac{1}{2}, \qquad \|B\|_{\sigma,\mathbb{C}}=\frac{1}{\sqrt{2}}, \qquad \|B\|_{*,\mathbb{R}}= 2, \qquad \|B\|_{*,\mathbb{C}} =\sqrt{2}.
\]
Furthermore, $B\in \mathsf{S}^3(\mathbb{R}^2)  \subseteq \mathsf{S}^3(\mathbb{C}^2)$ has a symmetric nuclear decomposition over $\mathbb{R}$ given by
\begin{equation}\label{RmindecS2}
B= \frac{2}{3}\biggl(\biggl[\frac{\sqrt{3}}{2}e_1+\frac{1}{2}e_2\biggr]^{\otimes 3}+ \biggl[-\frac{\sqrt{3}}{2}e_1+\frac{1}{2}e_2\biggr]^{\otimes 3} + (-e_2)^{\otimes 3} \biggr),
\end{equation}
and a symmetric nuclear decomposition over $\mathbb{C}$ given by
\begin{equation}\label{CmindecS2}
B=\frac{1}{\sqrt{2}}\biggl(\biggl[-\frac{1}{\sqrt{2}}e_2 +\frac{i}{\sqrt{2}} e_1\biggr]^{\otimes 3} +\biggl[-\frac{1}{\sqrt{2}}e_2 - \frac{i}{\sqrt{2}}e_1\biggr]^{\otimes 3}\biggr).
\end{equation}%
\end{lemma}
\begin{proof}
Since $B\in \mathsf{S}^3(\mathbb{R}^2)$, we may rely on  \eqref{banachthm} and \eqref{banachnuc} in Section~\ref{sec:banach} to calculate its spectral and nuclear norms over  $\mathbb{R}$ and $\mathbb{C}$. Set  $Y=2B$ for convenience.

Let $x=(x_1,x_2)^\mathsf{T}$ with  $|x_1|^2 + |x_2|^2=1$.  Then  $g(x_1,x_2)\coloneqq \langle Y, x^{\otimes 3}\rangle = 3x_1^2x_2 - x_2^3=x_2(3x_1^2-x_2^2)$.  Suppose first that $x_1,x_2\in\mathbb{R}$.  Then $x_1^2=1-x_2^2$ and the maximum of $g(x_1,x_2)=x_2(3-4x_2^2)$ over  $x_2 \in[0,1]$ is attained at $x_2=1/2$, $x_1 = \sqrt{3}/2$. Hence $\|Y\|_{\sigma,\mathbb{R}}=1$ and $\|B\|_{\sigma,\mathbb{R}}=1/2$.

Assume now that $x_1,x_2\in\mathbb{B}$.  Clearly, $|g(x_1,x_2)|\le |x_2|(3|x_1|^2+|x_2|^2)$.  Choose $x_2=-t$, $x_1=is$ 
where $s,t\ge 0$ and $s^2+t^2=1$.  Then the maximum of $g(x_1,x_2)=h(s,t)=t(3s^2+t^2)=t(3-2t^2)$ over $t\in [0,1]$ is $\sqrt{2}$, attained at $t=1/\sqrt{2} = s$.
Hence $\|B\|_{\sigma,\mathbb{C}}=1/\sqrt{2}$ and $\|Y\|_{\sigma,\mathbb{C}}=\sqrt{2}$.

That \eqref{exam2} is a nuclear decomposition over $\mathbb{R}$ and $\|B\|_{*,\mathbb{R}} = 2$ follows from Lemma~\ref{necsufmindec}  and the observation
\begin{equation}\label{Yident}
\langle Y, e_1\otimes  e_1\otimes  e_2\rangle = \langle Y, e_1\otimes  e_2\otimes  e_1\rangle =\langle Y, e_2\otimes  e_1\otimes  e_1\rangle =\langle Y, (-e_2)^{\otimes 3}\rangle =1=\|Y\|_{\sigma,\mathbb{R}}.
\end{equation}
That \eqref{RmindecS2} is a symmetric nuclear decomposition over $\mathbb{C}$ follows from Lemma~\ref{necsufmindec} and the observation
\[
\biggl\langle Y,\biggl[\frac{1}{\sqrt{2}}(-e_2+ie_1)\biggr]^{\otimes 3}\biggr\rangle =\biggl\langle Y,\biggl[\frac{1}{\sqrt{2}}(-e_2-ie_1)\biggr]^{\otimes 3}\biggr\rangle =\sqrt{2} = \| Y \|_{\sigma,\mathbb{C}}.
\]
This also shows that $\|B\|_{*,\mathbb{C}} = \sqrt{2}$.
\end{proof}
\begin{lemma}\label{3tenexamnorm1}
 Let $e_1, e_2 \in \mathbb{R}^2$ be the standard basis vectors. Define $C \in \mathbb{R}^{2 \times 2 \times 2} \subseteq \mathbb{C}^{2 \times 2 \times 2}$ by
\begin{equation}\label{exam1}
C=\frac{1}{\sqrt{3}}(e_1\otimes  e_1\otimes  e_2+e_1\otimes  e_2\otimes  e_1+e_2\otimes  e_1\otimes  e_1).
\end{equation}
Then \eqref{exam1} is a nuclear decomposition over $\mathbb{R}$, and
\begin{equation}\label{exam1specnucnrm}
\|C\|_{\sigma,\mathbb{R}}=\|C\|_{\sigma,\mathbb{C}}=\frac{2}{3},  \qquad \|C\|_{*,\mathbb{R}}=\sqrt{3}, \qquad \|C\|_{*,\mathbb{C}}=\frac{3}{2}.
\end{equation}
Furthermore, $C\in \mathsf{S}^3(\mathbb{R}^2) \subseteq \mathsf{S}^3(\mathbb{C}^2) $ has a symmetric nuclear decomposition over $\mathbb{R}$ given by
\begin{equation}\label{S1realsymmindec}
C=\frac{4}{3\sqrt{3}}\bigg(\biggl[\frac{\sqrt{3}}{2}e_1+\frac{1}{2}e_2\biggr]^{\otimes 3} + \biggl[-\frac{\sqrt{3}}{2}e_1+\frac{1}{2}e_2\biggr]^{\otimes 3} +\frac{1}{4}(-e_2)^{\otimes 3}\bigg),
\end{equation}
and a symmetric nuclear decomposition over $\mathbb{C}$ given by
\begin{multline}\label{compminsymS1}
C=\frac{3}{8}\biggl(\biggl[\sqrt{\frac{2}{3}}e_1+\frac{1}{\sqrt{3}}e_2\biggr]^{\otimes 3}+\biggl[-\sqrt{\frac{2}{3}}e_1+\frac{1}{\sqrt{3}}e_2\biggr]^{\otimes 3}\\
+\biggl[i\sqrt{\frac{2}{3}}e_1-\frac{1}{\sqrt{3}}e_2\biggr]^{\otimes 3}+\biggl[-i\sqrt{\frac{2}{3}}e_1-\frac{1}{\sqrt{3}}e_2\biggr]^{\otimes 3}\biggr).
\end{multline}
\end{lemma}
\begin{proof}
Since $C$ is a symmetric tensor, we may rely on  \eqref{banachthm} and \eqref{banachnuc} in Section~\ref{sec:banach} to calculate its spectral and nuclear norms over  $\mathbb{R}$ and $\mathbb{C}$.
Set $X=\sqrt{3}C$ for convenience.

Let $x=(x_1,x_2)^\mathsf{T}$.  Then $f(x_1,x_2)\coloneqq \frac{1}{3}\langle X, x^{\otimes 3}\rangle =x_1^2x_2$.  Clearly $\|X\|_{\sigma,\mathbb{R}}=\|X\|_{\sigma,\mathbb{C}}$ since all entries of $X$ are nonnegative.   For the maximum of $\lvert f(x) \rvert$ when $\lVert x \rVert = 1$, we may restrict to $x_1,x_2 \ge 0$, $x_1^2+x_2^2=1$.  Since the maximum of $f(x_1,x_2)=x_1^2\sqrt{1-x_1^2}$ over $x_1 \in [0,1]$ occurs at $x_1^2=2/3$, $x_2=1/\sqrt{3}$, we get the first two equalities in \eqref{exam1specnucnrm}.

By Lemma~\ref{necsufmindec} and \eqref{Yident} in the proof of Lemma~\ref{3tenexamnorm2}, \eqref{exam1} is a nuclear decomposition over $\mathbb{R}$. Hence $\|C\|_{*,\mathbb{R}}= \sqrt{3}$.

By Corollary~\ref{cor:snd}, $C$ has symmetric nuclear decompositions over both $\mathbb{R}$ and $\mathbb{C}$. That \eqref{S1realsymmindec} is a symmetric nuclear decomposition over $\mathbb{R}$ follows from Lemma~\ref{necsufmindec} and the observation that
\[
\langle Y, (-e_2)^{\otimes 3}\rangle = \biggl\langle Y, \biggl[\frac{\sqrt{3}}{2}e_1+\frac{1}{2}e_2\biggr]^{\otimes 3}\biggr\rangle = \biggl\langle Y, \biggl[-\frac{\sqrt{3}}{2}e_1+\frac{1}{2}e_2\biggr]^{\otimes 3}\biggr\rangle =1=\|Y\|_{\sigma,\mathbb{R}},
\]
where $Y$ is as defined in the proof of Lemma~\ref{3tenexamnorm2}.
Likewise, \eqref{compminsymS1} is a symmetric nuclear decomposition over $\mathbb{C}$ by Lemma~\ref{necsufmindec} and the observation that
\begin{align*}
\biggl\langle C, \biggl[\sqrt{\frac{2}{3}}e_1+\frac{1}{\sqrt{3}}e_2\biggr]^{\otimes 3} \biggr\rangle &= \biggl\langle C, \biggl[-\sqrt{\frac{2}{3}}e_1+\frac{1}{\sqrt{3}}e_2\biggr]^{\otimes 3}\biggr\rangle \\
&= \biggl\langle C, \biggl[i\sqrt{\frac{2}{3}}e_1-\frac{1}{\sqrt{3}}e_2\biggr]^{\otimes 3}\biggr\rangle = \biggl\langle C, \biggl[-i\sqrt{\frac{2}{3}}e_1-\frac{1}{\sqrt{3}}e_2\biggr]^{\otimes 3}\biggr\rangle =\|C\|_{\sigma,\mathbb{C}}.
\end{align*}
Since \eqref{compminsymS1} is a symmetric nuclear decomposition over $\mathbb{C}$, we obtain $\|C\|_{*,\mathbb{C}}=3/2$.
\end{proof}

%
%

Let $x=(x_1,\dots,x_n)^\mathsf{T}\in\mathbb{C}^n$.  Denote by $|x|\coloneqq (|x_1|,\dots,|x_n|)^\mathsf{T}$.  Then $x$ is called a nonnegative vector, denoted as $x\ge 0$, if $x=|x|$.
We will also use this notation for tensors in $\mathbb{C}^{n_1\times \dots \times n_d}$.  
\begin{lemma}\label{specnormnonneg}  Let $A\in \mathbb{C}^{n_1\times \dots \times n_d}$.
Then 
\[\|A\|_{\sigma, \mathbb{C}}\le \||A|\|_{\sigma, \mathbb{C}}, \quad \||A|\|_{\sigma,\mathbb{C}}=\||A|\|_{\sigma,\mathbb{R}}.\]
\end{lemma}
\begin{proof}  The triangle inequality yields  
\[
|\langle A, x_1 \otimes \dots \otimes x_d \rangle|\le \langle |A|, |x_1| \otimes \dots \otimes |x_d| \rangle.
\]
Recall that the Euclidean norm on $\mathbb{C}^n$ is an absolute norm, i.e., $\|x\|=\||x|\|$.  The definitions of $\|\cdot\|_{\sigma, \mathbb{C}}$ and $\|\cdot\|_{\sigma,\mathbb{R}}$
and the above inequality yields the result.
\end{proof}

A plausible nuclear norm analogue of the inequality $\|A\|_{\sigma, \mathbb{C}}\le \||A|\|_{\sigma, \mathbb{C}}$ is $\|A\|_{*, \mathbb{C}}\le \||A|\|_{*, \mathbb{C}}$. 
It is easy to show that this inequality holds in special cases (e.g.\ if $A$ is a hermitian positive semidefinite matrix) but it is false in general. For example, let
\[
A=\begin{bmatrix}1/\sqrt{2}& 1/\sqrt{2}\\ -1/\sqrt{2}&1/\sqrt{2}\end{bmatrix}.
\]
Then $\|A\|_*=2>\sqrt{2}=\||A|\|_*$.

\section{Nuclear $(p,q)$-norm of a matrix}\label{sec:matrix}

In this section, we study the special case where $d = 2$. Let $\| \cdot \|_p$ denote the $l^p$-norm on $\mathbb{R}^n$, i.e.,
\[
\|x\|_p=\left(\sum_{i=1}^n |x_i|^p\right)^{1/p}, \qquad \|x\|_\infty = \max \{ |x_1|,\dots,|x_n|\}.
\]
Recall that the dual norm  $\|\cdot\|_p^*=\|\cdot\|_{p^*}$ where $p^* \coloneqq p/(p-1)$, i.e., $1/p + 1/p^*=1$. 

The nuclear $(p,q)$-norm of a matrix $A \in \mathbb{R}^{m \times n}$ is
\begin{equation}\label{eq:npq}
\|A\|_{*,p,q} = \inf\Bigl\{\sum_{i=1}^r \lvert \lambda_i\rvert : A =  \sum_{i=1}^r \lambda_i u_{i}\otimes  v_{i}, \; \lVert u_{i} \rVert_p =\lVert v_{i} \rVert_q = 1, \; r \in \mathbb{N}\Bigr\}
\end{equation}
for any $p, q\in[1,\infty]$.  The spectral $(p,q)$-norm on $\mathbb{R}^{m \times n}$ is
\[
\|A\|_{\sigma,p,q} =  \max_{x, y \ne 0} \frac{y^\mathsf{T} A x}{\| x\|_p\| y\|_q} =   \max_{\| x\|_p = \| y\|_q = 1} y^\mathsf{T} A x
\]
for any $p, q\in[1,\infty]$.  The operator $(p,q)$-norm on $\mathbb{R}^{m \times n}$ is
\[
\|A\|_{p,q} =  \max_{x \ne 0} \frac{\|A x\|_q}{\| x\|_p} = \max_{\| x\|_p=1} \|A x\|_q
\]
for any $p, q\in[1,\infty]$. 
When $p = q$, we write
\[
\| \cdot \|_{p,p} = \| \cdot \|_{p}, \qquad \| \cdot \|_{\sigma, p,p} = \| \cdot \|_{\sigma, p}, \qquad \| \cdot \|_{*, p,p} = \| \cdot \|_{*, p},
\]
and call them the operator, spectral, nuclear $p$-norm respectively. The case $p = 2$ gives the usual spectral and nuclear norms.

It is well-known that  the operator $(p,q)$-norm and the spectral $(p,q)$-norm are identical:
\[
\|A\|_{\sigma,p,q}  = \|A\|_{p,q} \qquad\text{for all} \; A \in \mathbb{R}^{m \times n},
\]
and henceforth we will use the operator $(p,q)$-norm since it is the better known one.  It follows from $\|A x\|_q= \max_{\| y\|_{q^*}=1}  y^\mathsf{T} A x$ and $ y^\mathsf{T} A x = x^\mathsf{T} A^\mathsf{T} y$ that
\begin{equation}\label{eqnrmTdual}
\|A^\mathsf{T}\|_{q^*,p^*}=\|A\|_{p,q}.
\end{equation}

Equivalently, \eqref{eq:npq} may be written
\begin{equation}\label{eq:npq1}
\|A\|_{*,p,q} \coloneqq \min\Bigr\{\sum_{i=1}^r \|x_i\|_p\| y_i\|_{q} : A=\sum_{i=1}^r x_i\otimes y_i, \; r \in \mathbb{N}\Bigr\},
\end{equation}
or as the norm whose unit ball is the convex hull of  all ranks-one matrices $x\otimes y$, where $\|x\|_p\| y\|_q\le 1$. It is trivial to deduce from \eqref{eq:npq1} an analogue of \eqref{eqnrmTdual},
\begin{equation}\label{eqnucnrmTdual} 
\|A^\mathsf{T}\|_{*,q,p}=\|A\|_{*,p,q}.
\end{equation}

\begin{theorem}\label{dualnucopnorm}
The dual norm of the   operator $(p,q)$-norm is the nuclear $(q^*,p)$-norm on $\mathbb{R}^{m \times n}$, i.e.,
\[
\|A\|_{p,q}^*  = \|A\|_{*,q^*,p}
\]
for all $A \in \mathbb{R}^{m \times n}$ and all $p,q \in [1,\infty]$.
\end{theorem}
\begin{proof}
As in the proof of Corollary~\ref{cor:bnuclear}, the unit ball of the $(q^*,p)$-nuclear norm $\| \cdot \|_{*,q^*,p}$ on $\mathbb{R}^{m\times n}$ is the convex hull of
$\mathcal{E}=\{xy^\mathsf{T} :  \lVert x \rVert _{q^*}= \lVert y \rVert _{p}=1\}$.
Hence
\begin{align*}
\|A\|_{*,q^*,p}^* &=\max_{\|B\|_{*,q^*,p}\le 1}\operatorname{tr} (B^\mathsf{T}A)=\max_{xy^\mathsf{T}\in\mathcal{E}} \operatorname{tr} (yx^\mathsf{T}A)\\
&= \max_{\| x\|_{q^*}=\| y\|_p=1} x^\mathsf{T}A y = \max_{\| y\|_p= 1} \lVert A y\rVert_{q}=\|A\|_{p,q}.  \qedhere
\end{align*}
\end{proof}

It is well-known that the operator $(p,q)$-norm is NP-hard in many instances \cite{HO10, Ste05} notably:
\begin{enumerate}[\upshape (i)]
\item\label{i} $\| \cdot \|_{p,q}$ is NP-hard if $1\le q<p\le \infty$.  
\item  $\| \cdot \|_p$  is NP-hard if $p\ne 1,2,\infty$.
\end{enumerate}
The exceptional cases \cite{Ste05} are also well-known:
\begin{enumerate}[\upshape (i)]\setcounter{enumi}{2}
\item $\| \cdot \|_{p}$  is polynomial-time computable if $p = 1,2,\infty$. 
\item\label{iv} $\| \cdot \|_{p,q}$  is polynomial-time computable if $p=1$ and $1\le q\le \infty$, or  if $q = \infty$ and $1\le p \le \infty$.
\end{enumerate}
By \cite{FL1}, the computational complexity of norms and their dual norms are polynomial-time interreducible. So we obtain the following from Theorem~\ref{dualnucopnorm}.
\begin{enumerate}[\upshape (i)]\setcounter{enumi}{4}
\item $\| \cdot \|_{*,p,q}$ is NP-hard if $1\le p^*< q\le \infty$.  
\item $\| \cdot \|_{*,p^*,p}$  is NP-hard if $p\ne 1,2,\infty$.
\item $\| \cdot \|_{*,p^*,p}$  is polynomial-time computable if $p = 1,2,\infty$. 
\item\label{iv'} $\| \cdot \|_{*,p,q}$  is polynomial-time computable if $p=1$ and $1\le q\le \infty$, or  if $q = 1$ and $1\le p \le \infty$.
\end{enumerate}
In \eqref{iv} and \eqref{iv'}, we assume that the values of $p$ and $q$ are rational. 

In fact, as further special cases of \eqref{iv'}, the nuclear $(1,p)$-norms and $(p,1)$-norms have closed-form expressions, a consequence of the well-known
closed-form expressions for the operator $(1,p)$-norms and $(p,\infty)$-norms.
\begin{proposition}\label{comp1pnorm}    Let $e_1,\dots,e_n$ be the standard basis vectors in $\mathbb{R}^n$.  Let $A\in \mathbb{R}^{m\times n}$ and write
\[
A = [A_{\bullet 1},\dots,A_{\bullet n}] = \begin{bmatrix} A_{1\bullet}^\mathsf{T}\\ \vdots \\ A_{m\bullet}^\mathsf{T} \end{bmatrix},
\]
$A_{\bullet 1},\dots,A_{\bullet n} \in \mathbb{R}^m$ are the column vectors and $A_{1\bullet},\dots,A_{m\bullet} \in \mathbb{R}^n$ are the row vectors of $A$. Then
\begin{align}
\|A\|_{1,p}&=\max_{j=1,\dots,n} \|Ae_{j}\|_p = \max \{\|A_{\bullet 1}\|_p,\dots,\|A_{\bullet n}\|_p\}, \label{1pform1}\\
\|A\|_{p,\infty}&=\max_{i=1,\dots,m}\|A^\mathsf{T}e_{i}\|_{p^*}= \max \{\|A_{1\bullet}\|_{p^*},\dots,\|A_{m\bullet}\|_{p^*}\},\label{1pform2}\\
\|A\|_{*,1,p}&=\sum_{i=1}^m  \|A^\mathsf{T}e_{i}\|_{p} = \|A_{1\bullet}\|_{p} + \dots + \|A_{m\bullet}\|_{p}, \label{1pnuform1}\\
\|A\|_{*,p,1}&=\sum_{j=1} ^n\|Ae_{j}\|_{p} = \|A_{\bullet 1}\|_p + \dots + \|A_{\bullet n}\|_p, \label{1pnuform2}
\end{align}
for all $p \in[1,\infty]$.
\end{proposition}
\begin{proof}  Note that  $\mathcal{C} =\{x \in \mathbb{R}^n : \|x\|_1 \le1\}$ is the convex hull of $\{\pm e_{j} : j = 1,\dots,n\}$.
As $x \mapsto \|A x\|_p$ is a convex function on $\mathcal{C}$, we deduce that $\|A\|_{1,p}=\max_{x\in \mathcal{C}}\|A x\|_p=\max_{j=1,\dots, n} \lVert\pm Ae_{j}\rVert$.  Hence \eqref{1pform1} holds.  \eqref{1pform2} then follows from \eqref{eqnrmTdual} and \eqref{1pform1}. Now observe that
\[
\|A\|_{1,p^*}^*=\max_{\|B\|_{1,p^*}\le 1}\operatorname{tr} (B^\mathsf{T}A) =
\max_{\|Be_{j}\|_{p^*}\le 1}\left|\sum_{j=1}^n (B e_{j})^\mathsf{T}(Ae_{j})\right|=\sum_{j=1}^n \|Ae_{j}\|_{p}.
\]
Using Theorem~\ref{dualnucopnorm}, we obtain \eqref{1pnuform2}.  \eqref{1pnuform1} then follows from \eqref{eqnucnrmTdual}  and \eqref{1pnuform2}.
\end{proof}

The operator $(\infty,1)$-norm is NP-hard to compute by \eqref{i} but it has a well-known expression \eqref{eq:infty1} that arises in many applications. We will describe its dual norm,  the nuclear $\infty$-norm. In the following, we let
\begin{align*}
\mathbb{E}^n &\coloneqq \{\varepsilon = (\varepsilon_1,\dots,\varepsilon_n)^\mathsf{T} \in \mathbb{R}^n : \varepsilon_i=\pm1, \; i=1,\dots,n\},\\
\mathbb{E}^m \otimes \mathbb{E}^n &\coloneqq \{ E = (\varepsilon_{ij} )\in \mathbb{R}^{m\times n} : \varepsilon_{ij} = \pm 1,\; i=1,\dots,m, \; j =1,\dots,n,\; \operatorname{rank}(E) = 1\}.
\end{align*}
Note that $\# \mathbb{E}^n = 2^n$ and $\#  \mathbb{E}^m \otimes \mathbb{E}^n =  2^{m+n-1}$.
\begin{lemma}\label{lem:intyonenorm} 
Let $A\in\mathbb{R}^{m\times n}$. Then
\begin{equation}\label{eq:inftyp}
\|A\|_{\infty,p}=\max_{\varepsilon\in\mathbb{E}^n}\|A\varepsilon\|_p.
\end{equation}
In particular,
\begin{equation}\label{eq:infty1}
\lVert A\rVert_{\infty,1}=\max_{\varepsilon_{1},\dots,\varepsilon_{m},\delta_{1}%
,\dots,\delta_{n}\in\{-1,+1\}}\sum_{i=1}^{m}\sum_{j=1}^{n}a_{ij}%
\varepsilon_{i}\delta_{j},
\end{equation}
and its dual norm is
\begin{equation}\label{eq:*infty}
\|A\|_{*,\infty} =\min\Bigl\{ \sum_{i=1}^{mn} |\lambda_i | : A = \sum_{i=1}^{mn} \lambda_i E_i,\; E_1,\dots, E_{mn} \in\mathbb{E}^m \otimes \mathbb{E}^n \; \text{linearly independent}\Bigr\}.
\end{equation}
\end{lemma}
\begin{proof}
Observe that the convex hull of $\mathbb{E}^n$ is precisely the unit cube, i.e.,
\[
\operatorname{conv}(\mathbb{E}^n) = \{ x \in  \mathbb{R}^n : \|x\|_{\infty} \le 1\},
\]
giving us \eqref{eq:inftyp}. For $x\in \mathbb{R}^m$,  note  that $\|x\|_1=\max_{\varepsilon\in\mathbb{E}^m}\varepsilon^\mathsf{T}x$ and thus
\[
\|A\|_{\infty,1}=\max_{\delta\in \mathbb{E}^n} \|A\delta\|_1=\max_{\varepsilon\in \mathbb{E}^m,\; \delta\in \mathbb{E}^n,} \varepsilon^\mathsf{T}A\delta,
\]
giving us \eqref{eq:infty1}. It follows from Theorem~\ref{dualnucopnorm} that $\|\cdot\|_{\infty,1}^*=\|\cdot\|_{*,\infty,\infty}=\|\cdot\|_{*,\infty}$ and \eqref{eq:*infty} follows from Proposition~\ref{prop:convcombextpts}.
\end{proof} 

We have thus far restricted our discussions over $\mathbb{R}$. We may use similar arguments to show that \eqref{eqnrmTdual}, \eqref{eqnucnrmTdual}, Theorem~\ref{dualnucopnorm}, and  Proposition~\ref{comp1pnorm} all remain true over $\mathbb{C}$. In addition, \eqref{eqnrmTdual} and \eqref{eqnucnrmTdual} also hold if we have $A^*$ in place of $A^\mathsf{T}$.

Nevertheless for $A\in\mathbb{R}^{m\times n}$, the values of its operator $(p,q)$-norm over $\mathbb{R}$ and over $\mathbb{C}$ may be different; likewise for its nuclear $(p,q)$-norm. In fact, a classical result \cite{Tay} states that  $\|A\|_{p,q,\mathbb{C}}=\|A\|_{p,q,\mathbb{R}}$ for all $A\in \mathbb{R}^{m\times n}$  if and only if $p\le q$. We deduce the following analogue for nuclear $(p,q)$-norm using Theorem~\ref{dualnucopnorm}.
\begin{corollary}
$\|A\|_{*,p,q,\mathbb{C}}=\|A\|_{*,p,q,\mathbb{R}}$ for all $A\in \mathbb{R}^{m\times n}$  if and only if $q\le p^*$.
\end{corollary}

\section{Tensor nuclear norm is NP-hard}\label{sec:np}

The computational complexity of a norm and  that of its dual norm are polynomial-time interreducible \cite{FL1}. If a norm is polynomial-time computable, then so is its dual; if a norm is NP-hard to compute, then so is its dual. Consequently, computing the nuclear norm of a $3$-tensor over $\mathbb{R}$ is NP-hard since computing the spectral norm of a $3$-tensor over $\mathbb{R}$ is NP-hard \cite{HL13}. In fact, it is easy to extend to higher orders by simply invoking Proposition~\ref{prop:mult}.
\begin{theorem}\label{thm:np1}
The spectral and nuclear norms of $d$-tensors over $\mathbb{R}$ are NP-hard for any $d \ge 3$.
\end{theorem}

In this section, we will extend the NP-hardness of tensor spectral and nuclear norms to $\mathbb{C}$.  In addition, we will show that even the \emph{weak} membership problem is NP-hard, a \emph{stronger} claim than the membership problem being NP-hard (Theorem~\ref{thm:np1} refers to the membership problem). In the study of various tensor problems, it is sometimes the case that imposing certain special properties on the tensors makes the problems more tractable. Examples of such properties include: (i) even order, (ii) symmetric or Hermitian, (iii) positive semidefinite, (iv) nonnegative valued (we will define these formally later). We will show that computing the spectral or nuclear norm for tensors having all of the aforementioned properties remains an NP-hard problem.

Let $G=(V,E)$ be an undirected graph with vertex set $V\coloneqq \{1,\dots,n\}$ and edge set $E \coloneqq \bigl\{ \{i_k,j_k\} : k=1,\dots,m\bigr\}$.
Let $\kappa(G)$ be the clique number of $G$, i.e., the size of the largest clique in $G$, well-known to be NP-hard to compute \cite{Ka}.  Let $M_G$ be the adjacency matrix of $G$, i.e., $m_{ij}=1=m_{ji}$ if $\{i,j\}\in E$ and is zero otherwise.
Motzkin and Straus \cite{MS65} showed that
\begin{equation}\label{modMSthm}
\frac{\kappa(G)-1}{\kappa(G)}=\max_{x\in\Delta^n} x^\mathsf{T} M_G x,
\end{equation} 
where  $\Delta^n \coloneqq \{ x \in \mathbb{R}^n : x \ge 0, \; \|x \|_1 = 1\}$ is the probability simplex.
Equality is attained in \eqref{modMSthm} when $x$ is uniformly distributed on the largest clique.

We transform \eqref{modMSthm} into a problem involving $4$-tensors.  Let $x=y^{\circ 2}$, 
i.e., $x=(y_1^2,\dots,y_n^2)^\mathsf{T}$.   Then\footnote{By convention, we sum once over each edge; e.g.\ if $E =\{ \{1,2\}\}$, then $\sum_{\{i,j\} \in E} a_{ij} = a_{12}$, not $a_{12} + a_{21}$.}
\begin{equation}\label{4tensconv}
 x^\mathsf{T} M_G x = 2\sum_{\{i,j\}\in E} y_i^2y_j^2.
\end{equation}
For integers $1\le s< t\le n$, let $A_{st} = \bigl(a_{ijkl}^{(s,t)} \bigr)_{i,j,k,l=1}^{n} \in \mathbb{C}^{n \times n \times n \times n}$ be defined by
\[
a_{ijkl}^{(s,t)}=
\begin{cases}
1/2 & i = s,\; j = t,\; k = s,\; l =t,\\
1/2 & i = t,\; j = s,\; k = t,\; l =s,\\
1/2 & i = s,\; j = t,\; k = t,\; l =s,\\
1/2 & i = t,\; j = s,\; k = s,\; l =t,\\
0 & \text{otherwise}.
\end{cases}
\]
Observe that $A_{st}$ is \emph{not} a symmetric tensor but we have
\begin{equation}\label{eq:ast}
\langle A_{st},y\otimes y\otimes  y\otimes y\rangle = 2y_s^2 y_t^2.
\end{equation}

\begin{definition}
Let $A = (a_{ijkl})_{i,j,k,l=1}^{m,n,m,n}\in \mathbb{C}^{m \times n \times m \times n}$ be a  $4$-tensor.
We call it \textit{bisymmetric} if 
\[
a_{ijkl}=a_{klij} \quad \textrm{for all }i,k=1,\dots, m,\; j,l = 1, \dots, n,
\]
and \textit{bi-Hermitian} if 
\[
a_{ijkl}=\bar{a}_{klij} \quad \text{for all }i,k = 1, \dots, m, \; j,l = 1, \dots, n.  
\]
A bi-Hermitian tensor is said to be \textit{bi-positive semidefinite} if
\[
\sum_{i,j,k,l=1}^{m,n,m,n} a_{ijkl} x_{ij} \bar{x}_{kl} \ge 0 \quad \text{for all}\;X = (x_{ij}) \in \mathbb{C}^{m \times n}.
\]
\end{definition}
We may regard a $4$-tensor $A = (a_{ijkl})_{i,j,k,l=1}^{m,n,m,n}\in \mathbb{C}^{m \times n \times m \times n}$ as a matrix $M(A)\coloneqq [a_{(i,j), (k,l)}]\in \mathbb{C}^{mn\times mn}$, where $a_{(i,j),(k,l)} \coloneqq a_{ijkl}$. Then $A$ is bisymmetric, bi-Hermitian, or bi-positive semidefinite if and only if $M(A)$ is symmetric, Hermitian, or positive semidefinite.

Clearly bi-Hermitian and bisymmetric are the same notion over $\mathbb{R}$.
If $m = n$, a bisymmetric $4$-tensor is not necessarily a symmetric $4$-tensor although the converse is trivially true. However, if $m = n$, a real  bi-positive semidefinite  tensor $A \in \mathbb{R}^{n \times n \times n \times n}$ is clearly a positive semidefinite tensor in the usual sense, i.e.,
\[
\sum_{i,j,k,l=1}^{n,n,n,n} a_{ijkl} x_{i} x_{j} x_{k} x_{l} \ge 0 \quad \text{for all}\; x \in \mathbb{R}^n.
\]

\begin{lemma}
The tensor $A_{st} \in \mathbb{C}^{n \times n \times n \times n}$ is bi-Hermitian, bisymmetric, bi-positive semidefinite, and has all entries nonnegative. 
\end{lemma}
\begin{proof}
It follows from the way it is defined that $A_{st}$ is bi-Hermitian, bisymmetric, and nonnegative valued. It is positive semidefinite because
\[
\sum_{i,j,s,t=1}^{n,n,n,n} a_{ijkl}^{(s,t)} x_{ij} \bar{x}_{kl} = \frac{1}{2} (x_{st} + x_{ts})(\bar{x}_{st} + \bar{x}_{ts}) \ge 0 
\]
for all $X = (x_{ij}) \in \mathbb{C}^{n \times n}$.
\end{proof}

$M(A_{st})$ is evidently a nonnegative definite, rank-one matrix with trace one.  Those familiar with quantum information theory may note that $M(A_{st})$ represents a bipartite  density matrix \cite{FL2}. For any graph $G= (V, E)$, we define
\begin{equation}\label{defCG}
A_G\coloneqq \sum_{\{s,t\}\in E}A_{st} \in \mathbb{C}^{n \times n \times n \times n}.
\end{equation}
Then  $A_G$ is bi-Hermitian, bisymmetric, bi-positive semidefinite, and has all entries nonnegative.
Summing \eqref{eq:ast} over $\{s, t\} \in E$ gives
\begin{equation}\label{eq:ms3}
\langle A_{G},y\otimes y\otimes  y\otimes y\rangle = x^\mathsf{T} M_Gx,
\end{equation}
where $x = y^{\circ 2}$. Hence
\begin{equation}\label{eq:ms4}
\max_{\lVert y \rVert = 1} \langle A_{G},y\otimes y\otimes  y\otimes y\rangle = \max_{x\in\Delta^n}  x^\mathsf{T} M_G x = \frac{\kappa(G)-1}{\kappa(G)} .
\end{equation}
\begin{theorem}\label{genban}
Let $G$ be a simple undirected graph on $n$ vertices with $m$ edges. Let $A_G$ be defined as in \eqref{defCG}.
Then 
\begin{equation}\label{inftynormcTy+}
\|A_G\|_{\sigma, \mathbb{C}} \coloneqq \max_{0 \ne x,y,u,v\in\mathbb{C}^n} \frac{|\langle A_G, x\otimes y \otimes u\otimes v \rangle |}{\|x\|\|y\|\|u\|\|v\|}
= \max_{0 \ne y\in\mathbb{R}^n_+} \frac{\langle A_G, y\otimes y \otimes y\otimes y\rangle}{\|y\|^4} .
\end{equation}
Furthermore, we have
\begin{equation}\label{norminfeq}
\frac{\kappa(G)-1}{\kappa(G)} = \|A_G\|_{\sigma,\mathbb{C}}=\|A_G\|_{\sigma,\mathbb{R}}.
\end{equation}
\end{theorem}
If $A_G$ were a symmetric $4$-tensor as opposed to merely bisymmetric, then we may apply Banach's theorem \eqref{banachthm} to deduce that the maximum is attained at $x=y=u=v$ and thus \eqref{inftynormcTy+} would follow.  However $A_G$ is not symmetric and we may not invoke Banach's theorem.  Instead we will rely on the following lemma, which may be of independent interest.
\begin{lemma}\label{partsymlem}  Let $A=(a_{ijkl})\in \mathbb{C}^{m \times n \times m \times n}$. If $M(A) \in \mathbb{C}^{mn \times mn}$ is Hermitian positive semidefinite, then
\[
\|A\|_{\sigma , \mathbb{C}}=\max_{0 \ne x\in \mathbb{C}^m, \; 0 \ne y \in \mathbb{C}^n} \frac{\langle A, x\otimes y \otimes \bar{x}\otimes \bar{y} \rangle}{\|x\|^2\|y\|^2}.
\]
\end{lemma}
\begin{proof}
Let $M = M(A)$. Then $M$ is a Hermitian positive semidefinite matrix. 
Cauchy--Schwarz applied to the sesquilinear form $\bar{w}^\mathsf{T} M z$ gives 
\[
|\bar{w}^\mathsf{T} Mz |\le\sqrt{\bar{z}^\mathsf{T} M z}\sqrt{\bar{w}^\mathsf{T} M w}\le \max(\bar{z}^\mathsf{T} M z ,\bar{w}^\mathsf{T} M w ).
\]
Let $ z = \operatorname{vec}( x\otimes y)$ and  $w = \operatorname{vec}(  \bar{u}\otimes \bar{v}) \in \mathbb{C}^{mn}$ and observe that
\[
|\langle A, x\otimes y \otimes u \otimes v \rangle| =  |\bar{w}^\mathsf{T} M z| \le  \max(\langle A, x\otimes y \otimes \bar{x}\otimes \bar{y} \rangle,\langle A,  \bar{u}\otimes \bar{v}  \otimes u\otimes v \rangle),
\]
from which the required equality follows upon taking $\max$ over unit vectors.
\end{proof}

\begin{proof}[Proof of Theorem~\ref{genban}]
We apply Lemma~\ref{partsymlem} to $A_G$ and note that we may take our maximum over $\mathbb{R}^n_+$  since $A_G$ is nonnegative valued.
\[
\|A_G \|_{\sigma,\mathbb{C}} =\max_{0 \ne x,y,u,v\in\mathbb{R}^n_+} \frac{\langle A_G, x\otimes y \otimes u\otimes v \rangle}{\|x\|\|y\|\|u\|\|v\|}
=\max_{0 \ne x, y \in \mathbb{R}^n_+} \frac{\langle A_G, x\otimes y \otimes x \otimes y \rangle}{\|x\|^2\|y\|^2}.
\]  
Since $2\langle A_{st}, x\otimes y \otimes x \otimes y \rangle =(x_sy_t+x_ty_s)^2$, we may use Cauchy--Schwarz to see that
\[
(x_sy_t+x_ty_s)^2\le 4 \frac{(x_s^2+y_s^2)}{2} \times \frac{(x_t^2+y_t^2)}{2}.
\]
If we do a change-of-variables $a_s=\sqrt{(x_s^2+y_s^2)/2}$ for $s=1,\dots, n$, we obtain
\[
\langle A_{st}, x\otimes y \otimes x \otimes y \rangle \le 2a_s^2 a_t^2=\langle A_{st}, a\otimes a \otimes a \otimes a \rangle.
\]
Upon summing over $\{s,t\} \in E$, we get
\[
\langle A_G, x\otimes y \otimes x \otimes y \rangle \le \langle A_G, a\otimes a \otimes a \otimes a \rangle,
\]
where the left-hand side follows from \eqref{defCG} and the right-hand side follows from \eqref{4tensconv} and \eqref{eq:ms3}.
The last inequality gives us \eqref{inftynormcTy+} easily. We then get \eqref{norminfeq} from \eqref{eq:ms4} and \eqref{inftynormcTy+}.
\end{proof}

In the following, we let  $\mathbb{Q}_\mathbb{F}$ be the field of rational numbers $\mathbb{Q}$ if $\mathbb{F} = \mathbb{R}$ and the field of Gaussian rational numbers $ \mathbb{Q}[i] \coloneqq \{ a + bi : a, b \in \mathbb{Q}\}$ if $\mathbb{F} = \mathbb{C}$. As is customary, we will restrict our problem inputs to $\mathbb{Q}_\mathbb{F}$ to ensure that they may be specified in finitely many bits. We refer the reader to \cite[Definitions~2.1 and 4.1]{FL1} for the formal definitions of the weak membership problem and the approximation problem.

Computing the clique number of a graph is an NP-hard problem \cite{Ka} and so the identity  \eqref{norminfeq}  implies that the computing the spectral norm of $A_G$ is NP-hard over both $\mathbb{R}$ and $\mathbb{C}$.  Since the clique numberh is an integer, it is also NP-hard to approximate the spectral norm to arbitrary accuracy.
\begin{theorem}\label{NPhardapprx4tensor}
Let $\delta > 0$ be rational and $A \in \mathbb{Q}_\mathbb{F}^{n \times n \times n \times n}$ be bi-Hermitian, bi-positive semidefinite, and nonnegative-valued. Computing an approximation $\omega(A) \in \mathbb{Q}$ such that
\[
\lVert A \rVert_{\sigma, \mathbb{F}} - \delta < \omega(A) < \lVert A \rVert_{\sigma, \mathbb{F}} + \delta
\]
is an NP-hard problem for both $\mathbb{F} = \mathbb{R}$ and $\mathbb{C}$.
\end{theorem}

For any $\delta>0$ and any convex set with nonempty interior $K \subseteq \mathbb{F}^n$,  we define
\[
S(K,\delta) \coloneqq \bigcup_{x\in K} B(x,\delta)\quad \text{and}\quad S(K,-\delta) \coloneqq \{x\in K: B(x,\delta)\subseteq K\},
\]
where $B(x,\delta)$ is the $\delta$-ball centered at $x$ with respect to the Hilbert--Schmidt norm in $\mathbb{F}^n$.
Using \cite[Theorem~4.2]{FL1}, we deduce the NP-hardness of the weak membership problem from Theorem~\ref{NPhardapprx4tensor}.
\begin{corollary}\label{cor:hardCR}
Let $K$ be the spectral norm unit ball in $\mathbb{F}^{n\times n \times n \times n}$ and  $0<\delta \in \mathbb{Q}$.
Given $A \in \mathbb{Q}_\mathbb{F}^{n \times n \times n \times n}$ that is bi-Hermitian, bi-positive semidefinite, and nonnegative-valued, deciding whether $A\in S(K,\delta)$ or $x\notin S(K,-\delta)$ is an NP-hard problem  for both $\mathbb{F} = \mathbb{R}$ and $\mathbb{C}$.
\end{corollary}
It then follows from \cite[Theorem~3.1]{FL1} and the duality of spectral and nuclear norm that Corollary~\ref{cor:hardCR} also holds true for nuclear norm of $4$-tensors.
\begin{corollary}\label{cor:hardCR1}
Let $K$ be the nuclear norm unit ball in $\mathbb{F}^{n\times n \times n \times n}$ and  $0<\delta \in \mathbb{Q}$.
Given $A \in \mathbb{Q}_\mathbb{F}^{n \times n \times n \times n}$ that is bi-Hermitian, bi-positive semidefinite, and nonnegative-valued, deciding whether $A\in S(K,\delta)$ or $x\notin S(K,-\delta)$ is an NP-hard problem  for both $\mathbb{F} = \mathbb{R}$ and $\mathbb{C}$.
\end{corollary}
Using \cite[Theorem~4.2]{FL1} a second time, we may deduce the nuclear norm analogue of Theorem~\ref{NPhardapprx4tensor}.
\begin{corollary}\label{cor:hardCR2}
Let $\delta > 0$ be rational and $A \in \mathbb{Q}_\mathbb{F}^{n \times n \times n \times n}$ be bi-Hermitian, bi-positive semidefinite, and nonnegative-valued. Computing an approximation $\omega(A) \in \mathbb{Q}$ such that
\[
\lVert A \rVert_{\sigma, \mathbb{F}} - \delta < \omega(A) < \lVert A \rVert_{\sigma, \mathbb{F}} + \delta
\]
is an NP-hard problem for both $\mathbb{F} = \mathbb{R}$ and $\mathbb{C}$.
\end{corollary}
As we did for Theorem~\ref{thm:np1}, we may use  Corollaries~\ref{cor:hardCR} and \ref{cor:hardCR1} along with Proposition~\ref{prop:mult} to deduce a complex analogue of Theorem~\ref{thm:np1}.
\begin{theorem}\label{thm:np2}
The spectral and nuclear norms of $d$-tensors over $\mathbb{C}$ are NP-hard for any $d \ge 4$.
\end{theorem}

\section{Polynomial-time approximation bounds}\label{sec:approx}

Assuming that $\mathit{P} \ne \mathit{NP}$, then by Corollaries~\ref{cor:hardCR} and \ref{cor:hardCR1}, one cannot approximate the spectral and nuclear norms of $d$-tensors to arbitrary accuracy in polynomial time. In this section, we will discuss some approximation bounds for spectral and nuclear norms that are computable in polynomial time. 

The simplest polynomial-time computable bounds for the spectral and nuclear norms are those that come from the equivalence of norms in finite-dimensional spaces. The following lemma uses the Hilbert--Schmidt norm but any other H\"{o}lder $p$-norms \cite{Lim13},
\[
\| A \|_{H,p} \coloneqq  \Bigl( \sum_{i_1, \dots, i_d=1}^{n_1, \dots, n_d} \lvert a_{i_1 \cdots i_d}\rvert^p \Bigr)^{1/p},
\] 
where $p \in [1,\infty]$, which are all polynomial-time computable, may also serve the role.
\begin{lemma}\label{tenspecnrmest}  Let $A \in \mathbb{F}^{n_1\times \dots \times n_d}$.
Then 
\[
\frac{1}{\sqrt{ n_1 \cdots n_d}}\|A\|\le \|A\|_{\sigma}\le \|A\| \qquad \text{and}\qquad  \|A\|\le \|A\|_{*}\le \sqrt{ n_1 \cdots n_d}\|A\| .
\]
\end{lemma}
\begin{proof} 
We start with the bounds for the spectral norm. Clearly $\|A\|_{\sigma}\le \|A\|$.  Let $A=(a_{i_1 \cdots i_d})$ and set $\|A \|_{H,\infty}=\max\{|a_{i_1 \cdots i_d}| : i_k =1,\dots, n_k,\;k=1,\dots,d\}$.
Clearly, $\|A\|\le \sqrt{ n_1\cdots n_d} \; \|A \|_{H,\infty}$.  Note that $a_{i_1 \cdots i_d} = \langle A, e_{i_1} \otimes \dots \otimes e_{i_d}\rangle$ where $e_{i_k}$ are standard basis vectors 
in $\mathbb{F}^{n_k}$. In particular $\|A \|_{H,\infty}  = \lvert \langle A, u_1 \otimes \dots \otimes u_d\rangle \rvert$ for some unit vectors $u_1,\dots,u_d$ and thus $\|A \|_{H,\infty} \le \|A\|_{\sigma}$ by \eqref{def2tenspnrm}.   The corresponding inequalities for the nuclear norm follows from it being a dual norm.
\end{proof}

One downside of universal bounds like those in Lemma~\ref{tenspecnrmest} is that they necessarily depend on the dimension of the ambient space. We will now construct tighter polynomial-time computable bounds for the spectral and nuclear norms of $3$-tensors that depend only on the `intrinsic dimension' of the specific tensor we are approximating. The \emph{multilinear rank} \cite{DSL} of  a $3$-tensor $A\in\mathbb{F}^{m\times n\times p}$ is the $3$-tuple
$\mu\operatorname{rank}(A)\coloneqq (r_{1},r_{2},r_{3})$ where
\begin{align*}
r_{1} &=\dim\operatorname{span}_{\mathbb{F}}\{A_{1\bullet\bullet },\dots,A_{m\bullet\bullet}\},\\
r_{2} &=\dim\operatorname{span}_{\mathbb{F}}\{A_{\bullet1\bullet},\dots,A_{\bullet n\bullet}\},\\
r_{3} &=\dim\operatorname{span}_{\mathbb{F}}\{A_{\bullet\bullet 1},\dots,A_{\bullet\bullet p}\}.
\end{align*}
Here $A_{i\bullet\bullet}=(a_{ijk})_{j,k=1}^{n,p}\in\mathbb{F}^{n\times p}$,
$A_{\bullet j\bullet}=(a_{ijk})_{i,k=1}^{m,p}\in\mathbb{F}^{m\times p}$,
$A_{\bullet\bullet k}=(a_{ijk})_{i,j=1}^{m,p}\in\mathbb{F}^{m\times n}$ are
`matrix slices' of the $3$-tensor --- the analogues of the row and column
vectors of a matrix. This was due originally to Hitchcock \cite{Hi2}, a special case ($2$-plex rank) of his \textit{multiplex rank}.

We define the \emph{flattening maps} along the $1$st, $2$nd, and $3$rd index by
\[
\flat_{1}:\mathbb{F}^{m\times n\times p}\rightarrow\mathbb{F}^{m \times
np},\quad\flat_{2}:\mathbb{F}^{m\times n\times p}\rightarrow\mathbb{F}%
^{n\times mp},\quad\flat_{3}:\mathbb{F}^{m\times n\times p}\rightarrow
\mathbb{F}^{p\times mn}%
\]
respectively. Intuitively, these take a $3$-tensor $A\in\mathbb{F}^{m\times n\times p}$ and `flatten' it in three different ways to yield
three matrices. Instead of giving precise but cumbersome formulae, it suffices to illustrate these simple maps with an example: Let%
\[
A=\left[
\begin{array}
[c]{r@{\quad}r@{\quad}r}
a_{111} & a_{121} & a_{131}\\
a_{211} & a_{221} & a_{231}\\
a_{311} & a_{321} & a_{331}\\
a_{411} & a_{421} & a_{431}%
\end{array}
\right\vert \!\left.
\begin{array}
[c]{r@{\quad}r@{\quad}r}
a_{112} & a_{122} & a_{132}\\
a_{212} & a_{222} & a_{232}\\
a_{312} & a_{322} & a_{332}\\
a_{412} & a_{422} & a_{432}%
\end{array}
\right]  \in\mathbb{F}^{4\times3\times2},
\]
then%
\begin{gather*}
\flat_{1}(A)=\left[
\begin{array}
[c]{r@{\quad}r@{\quad}r@{\quad}r@{\quad}r@{\quad}r}
a_{111} & a_{112} & a_{121} & a_{122} & a_{131} & a_{132}\\
a_{211} & a_{212} & a_{221} & a_{222} & a_{231} & a_{232}\\
a_{311} & a_{312} & a_{321} & a_{322} & a_{331} & a_{332}\\
a_{411} & a_{412} & a_{421} & a_{422} & a_{431} & a_{432}
\end{array}\right] \in\mathbb{F}^{4\times6},\\
\flat_{2}(A)= \left[\begin{array}{c@{\quad}c@{\quad}c@{\quad}c@{\quad}c@{\quad}c@{\quad}c@{\quad}c}
a_{111} & a_{112} & a_{211} & a_{212} & a_{311} & a_{312} & a_{411} & a_{412}\\
a_{121} & a_{122} & a_{221} & a_{222} & a_{321} & a_{322} & a_{421} & a_{422}\\
a_{131} & a_{132} & a_{231} & a_{232} & a_{331} & a_{332} & a_{431} & a_{432}
\end{array}\right]
  \in\mathbb{F}^{3\times8},\\
\flat_{3}(A)= \left[\begin{array}{c@{\quad}c@{\quad}c@{\quad}c@{\quad}c@{\quad}c@{\quad}c@{\quad}c@{\quad}c@{\quad}c@{\quad}c@{\quad}c}
a_{111} & a_{121} & a_{131} &a_{211} & a_{221} & a_{231}  & a_{311} & a_{321} & a_{331} & a_{411} & a_{421} & a_{431}  \\
a_{112} & a_{122} & a_{132} &a_{212} & a_{222} & a_{232}  & a_{312} & a_{322} & a_{332} & a_{412} & a_{422} & a_{432}
\end{array}\right]
  \in\mathbb{F}^{2\times12}.
\end{gather*}
It follows immediately from  definition that the multilinear rank $\mu\operatorname{rank}(A) = (r_{1},r_{2},r_{3})$ is given by%
\[
r_{1}=\operatorname{rank}(\flat_{1}(A)),\quad r_{2}=\operatorname{rank}%
(\flat_{2}(A)),\quad r_{3}=\operatorname{rank}(\flat_{3}(A)),
\]
where rank here is the usual matrix rank of the matrices $\flat_{1}(A),\flat_{2}(A),\flat_{3}(A)$. Although we will have no use for it, a recently popular definition of tensor nuclear norm is as the arithmetic mean of the (matrix) nuclear norm of the flattenings:
\[
\lVert A \rVert_\flat = \frac{1}{3}(\lVert \flat_{1}(A) \rVert_* + \lVert \flat_{2}(A) \rVert_* + \lVert \flat_{3}(A) \rVert_*).
\]

We first provide alternative characterizations for the spectral and nuclear norms of a $3$-tensor.
\begin{lemma}
Let $A \in \mathbb{F}^{m \times n  \times p}$. Then
\begin{align}
\|A \|_\sigma &= \max \biggl\{ \frac{|\langle A, x \otimes M \rangle|}{\|x\| \|M \|_\sigma} : 0 \ne x \in \mathbb{F}^m,\; 0 \ne M \in \mathbb{F}^{n \times p} \biggr\}, \label{anidlsnrm}\\
\|A\|_* &= \min\Bigl\{\sum_{i=1}^r \|x_i\|\|M_i\|_* : A = \sum_{i=1}^r x_i\otimes M_i, \; x_i \in \mathbb{F}^m, \; M_i\in \mathbb{F}^{n\times p}, \; r\in \mathbb{N} \Bigr\}. \label{anidl1nrm}
\end{align}
Furthermore there is a decomposition of $A$ that attains the minimum in \eqref{anidl1nrm} where $x_1\otimes M_1,\dots, x_r\otimes M_r$ are linearly independent.
\end{lemma}
\begin{proof}
If we set $M = y \otimes z$, then \eqref{anidlsnrm} becomes \eqref{def1tenspnrm}. So the maximum in  \eqref{anidlsnrm} is at least as large as the maximum in \eqref{def1tenspnrm}. On the other hand, the \textsc{svd} of $M$ shows that $\lVert M \rVert_\sigma = \lVert  \sigma_1 u_1 \otimes v_1 \rVert_\sigma$ where $\sigma_1$, $u_1$, $v_1$ are the largest singular values/vectors of $M$ and so we may always replace any $M$  in \eqref{anidlsnrm} that is not rank-one by its best rank-one approximation $\sigma_1 u_1 \otimes v_1$, giving us \eqref{def1tenspnrm}.

If we set $M_i=y_i\otimes z_i$, $i=1,\dots,r$, then \eqref{anidl1nrm} reduces to \eqref{def1tennrm}.  So the minimum in \eqref{def1tennrm} is not more than the minimum in \eqref{anidl1nrm}.  On the other hand, we may write each $M_i$ as a sum of rank-one matrices, in which case \eqref{anidl1nrm} reduces to \eqref{def1tennrm}. The existence of a decomposition that attains \eqref{anidl1nrm} follows from the same argument that we used in the proof of Proposition~\ref{prop:tnn}.  The linear independence of $x_1\otimes M_1,\dots x_r\otimes M_r$ follows from Proposition~\ref{prop:convcombextpts}.
\end{proof}

\begin{lemma}\label{lem:mrank}
Let $A \in \mathbb{F}^{m \times n  \times p}$ with  $\mu\operatorname{rank}(A) = (r_{1},r_{2},r_{3})$. If the decomposition
\[
A = \sum_{i=1}^r x_i\otimes M_i,
\]
attains \eqref{anidl1nrm}, then for all $i=1,\dots,r$,
\begin{equation}\label{eq:Mrr}
\operatorname{rank} M_i\le \min(r_2, r_3 ).
\end{equation}
\end{lemma}
\begin{proof}
Suppose $\mathbb{F} = \mathbb{R}$; the proof for $\mathbb{C}$ is similar except that we have unitary transformations in place of orthogonal ones.
Using any one of the multilinear rank decompositions \cite{Lim13}, we may reduce $A \in \mathbb{R}^{m \times n \times p}$ to a tensor  $U \in \operatorname{O}(m)$, $V \in \operatorname{O}(n)$, $W \in \operatorname{O}(p)$ such that
\[
A =  (U,V,W) \cdot C 
\]
where $C  \in \mathbb{R}^{m \times n \times p}$ is such that $c_{ijk} = 0$ if $i > r_1$, $j>r_2$, or $k > r_3$.
So we have
\[
 (U, V, W) \cdot C = \sum_{\ell =1}^r x_\ell \otimes M_\ell ,
\]
and applying the multilinear transform  $(U^\mathsf{T}, V^\mathsf{T}, W^\mathsf{T})$  to both sides, we get
\[
C = \sum_{\ell =1}^r (U^\mathsf{T}x_\ell ) \otimes( VM_\ell W^\mathsf{T}).
\]
Let $\ell = 1,\dots,r$. Let us partition $\widetilde{x}_\ell  = U^\mathsf{T}x_\ell \in \mathbb{R}^{m}$ and $\widetilde{M}_\ell  = VM_\ell W^\mathsf{T} \in \mathbb{R}^{n \times p}$ into
\begin{align*}
\widetilde{x}_\ell  &=\begin{bmatrix}y_\ell \\ z_\ell \end{bmatrix}, \quad y_\ell  \in \mathbb{R}^{r_1},\; z_\ell \in \mathbb{R}^{m -r_1},\\
\widetilde{M}_\ell  &=
\begin{bmatrix}
J_\ell  &K_\ell \\
L_\ell  &N_\ell 
\end{bmatrix}, \quad  J_\ell \in \mathbb{R}^{r_2 \times r_3},\; K_\ell  \in \mathbb{R}^{r_2 \times (p - r_3) }, \; L_\ell  \in \mathbb{R}^{(n -r_2)  \times r_3}, \; N_\ell  \in \mathbb{R}^{(n-r_2) \times (p -r_3)}.
\end{align*}
Now set
\[
x_\ell'=\begin{bmatrix}y_\ell \\0\end{bmatrix}, \qquad
M_\ell'=\begin{bmatrix}
J_\ell &0\\0&0\end{bmatrix}.
\]
As $c_{ijk}=0$ if $i > r_1$, $j>r_2$, or $k>r_3$, it follows that
\[
C =\sum_{\ell =1}^r x_\ell'\otimes M_\ell'.
\]
Since orthogonal matrices preserve Hilbert--Schimdt and nuclear norms,  $\|x_i\| =\|\widetilde{x}_i\| \ge \|x_i'\|$ and $\|M_\ell \|_* = \|\widetilde{M}_\ell \|_*\ge \|M_\ell '\|_*$ and so
\[
\sum_{\ell =1}^r \|x_\ell \|\|M_\ell \|_*\ge \sum_{\ell =1}^r \|x_\ell '\|\|M_\ell '\|_*.
\]
Clearly $\operatorname{rank} M_\ell '\le \min(r_2, r_3 )$.
\end{proof}

By Definition~\ref{def1tennrm},
\[
\|\flat_1(A)\|_*=\min\Bigl\{\sum_{i=1}^r \|x_i\|\|M_i\| : A = \sum_{i=1}^r x_i\otimes M_i, \; x_i \in \mathbb{F}^m, \; M_i\in \mathbb{F}^{n\times p}, \; r\in \mathbb{N} \Bigr\},
\]
and since any matrix satisfies
\[
\|M_i\|\le \|M_i\|_*\le \sqrt{\operatorname{rank} M_i}\|M_i\|_{\sigma},
\]
using \eqref{anidl1nrm} and \eqref{eq:Mrr}, we obtain
\begin{equation}\label{eq:bdnn}
\|\flat_1(A)\|_*\le \|A\|_*\le \sqrt{\min(r_2(A), r_3(A) )}\|\flat_1(A)\|_*.
\end{equation}
From \eqref{eq:bdnn}, we deduce the corresponding bounds for its dual norm, 
\[
\|\flat_1(A)\|_\sigma\ge \|A\|_\sigma\ge \frac{1}{\sqrt{\min(r_2(A), r_3(A))}}\|\flat_1(A)\|_\sigma.
\]
Moreover, we may deduce analogous inequalities in terms of $\flat_2(A)$ and $\flat_3(A)$. We assemble these to get  the bounds in the following theorem.
\begin{theorem}\label{thm:approx}
Let $A \in \mathbb{F}^{m \times n \times p}$ with $\mu\operatorname{rank}(A) = (r_{1},r_{2},r_{3})$. Then
{\footnotesize
\[
\max \biggl\{ \frac{\|\flat_1(A)\|_\sigma}{\sqrt{\min(r_2, r_3)}},\frac{\|\flat_2(A)\|_\sigma}{\sqrt{\min(r_1, r_3)}}, \frac{\|\flat_3(A)\|_\sigma}{\sqrt{\min(r_1, r_2)}}\biggr\}
 \le \|A\|_{\sigma}\le \min \{ \|\flat_1(A)\|_\sigma,  \|\flat_2(A)\|_\sigma, \|\flat_3(A)\|_\sigma \}
\]}%
and
{\footnotesize
\begin{multline*}
\max\{ \|\flat_1(A)\|_*,  \|\flat_2(A)\|_*,  \|\flat_3(A)\|_* \}\le \|A\|_* \\
\le \min \bigl\{\sqrt{\min(r_2, r_3 )}\|\flat_1(A)\|_*, \sqrt{\min(r_1, r_3 )}\|\flat_2(A)\|_*, \sqrt{\min(r_1, r_2 )}\|\flat_3(A)\|_*\bigr\}.
\end{multline*}}%
\end{theorem}
Note that both upper and lower bounds are computable in polynomial time. Clearly, we may extend Theorem~\ref{thm:approx} to any $d > 3$ simply by flattening along $d$ indices. 

\section*{Acknowledgment}

We thank Harm~Derksen and Jiawang~Nie for enormously helpful discussions. We thank Li~Wang for help with numerical experiments that suggested that $\|C\|_{*,\mathbb{C}} = \sqrt{2}$ in Lemma~\ref{3tenexamnorm2}.
SF's work is partially supported by NSF DMS-1216393.  
LH's work is partially supported by AFOSR FA9550-13-1-0133, DARPA D15AP00109, NSF IIS 1546413, DMS 1209136, DMS 1057064.

\bibliographystyle{plain}

\end{document}